\documentclass[11pt,english]{article}
\usepackage[latin9]{inputenc}
\usepackage{float}
\usepackage{amsmath}
\usepackage{amsthm}
\usepackage{amssymb}
\usepackage[numbers]{natbib}

\makeatletter

\providecommand{\tabularnewline}{\\}

  \theoremstyle{plain}
  \newtheorem{thm}{\protect\theoremname}[section]
  \theoremstyle{plain}
  \newtheorem{question}{\protect\questionname}
  \theoremstyle{definition}
  \newtheorem{example}{\protect\examplename}[section]
  \theoremstyle{remark}
  \newtheorem{rem}{\protect\remarkname}[section]
  \theoremstyle{plain}
  \newtheorem{lem}{\protect\lemmaname}[section]
  \theoremstyle{plain}
  \newtheorem{cor}{\protect\corollaryname}[section]
  \theoremstyle{plain}
  \newtheorem{prop}{\protect\propositionname}[section]
  \theoremstyle{definition}
  \newtheorem{defn}{\protect\definitionname}[section]

\usepackage[all,2cell]{xy}  \UseAllTwocells \SilentMatrices
\usepackage{hyperref}
\usepackage{amsmath}
\usepackage{slashed}
\usepackage{graphicx}

\usepackage[misc,geometry]{ifsym} 

\hypersetup{colorlinks=true, linkcolor=blue, citecolor=blue, filecolor=blue, pagecolor=blue, urlcolor=red, pdfauthor={}, pdftitle={}}

\makeatother

\usepackage{babel}
  \providecommand{\definitionname}{Definition}
  \providecommand{\examplename}{Example}
  \providecommand{\lemmaname}{Lemma}
  \providecommand{\propositionname}{Proposition}
  \providecommand{\questionname}{Question}
  \providecommand{\remarkname}{Remark}
\providecommand{\corollaryname}{Corollary}
\providecommand{\theoremname}{Theorem}

\begin{document}

\title{Towards Axiomatization and General Results on Strong Emergence Phenomena
Between Lagrangian Field Theories}

\author{Yuri Ximenes Martins\footnote{yurixm@ufmg.br (corresponding author)}\,  and  Rodney Josu\'e Biezuner\footnote{rodneyjb@ufmg.br} }
\maketitle
\noindent \begin{center}
\textit{Departamento de Matem\'atica, ICEx, Universidade Federal de Minas Gerais,}  \\  \textit{Av. Ant\^onio Carlos 6627, Pampulha, CP 702, CEP 31270-901, Belo Horizonte, MG, Brazil}
\par\end{center}
\begin{abstract}
In this paper we propose a formal definition of what is a strong emergence
phenomenon between two parameterized field theories and we present
sufficient conditions ensuring the existence of such phenomena between
two given parameterized Lagrangian field theories. More precisely,
we prove that in an Euclidean background, typical parameterized kinetic
theories emerge from any elliptic multivariate polynomial theories.
Some concrete examples are given and a connection with the phenomenon
of gravity emerging from noncommutativity is made.
\end{abstract}

\section{Introduction \label{sec_introduction}}

$\quad\;\,$The term \emph{emergence phenomenon} has been used for
years in many different contexts. In each of them, Emergence Theory
is the theory which studies those kinds of phenomena. E.g, we have
versions of it in Philosophy, Art, Chemistry and Biology \citep{0,1,2}.
The term is also used many times in Physics, with different meanings
(for a review on the subject, see \citep{3,4}. For an axiomatization
approach, see \citep{5}). This reveals that the concept of emergence
phenomenon is very general and therefore difficult to formalize. Nevertheless,
we have a clue of what it really is: when looking at all those presentations
we see that each of them is about describing a system in terms of
some other system, possibly in different \emph{scales}. Thus, an emergence
phenomenon is about a relation between two different systems, the
\emph{emergence relation}, and a system emerges from another when
it (or at least part of it) can be recovered in terms of the other
system, which is presumably more fundamental, at least in some scale.
The different emergence phenomena in Biology, Philosophy, Physics,
and so on, are obtained by fixing in the above abstract definition
a meaning for system, scale, etc.

Notice that, in this approach, in order to talk about emergence we
need to assume that to each system of interest we have assigned a
\emph{scale}. In Mathematics, scales are better known as \emph{parameters}.
So, emergence phenomena occur between some kinds of \emph{parameterized
systems}. This kind of assumption (that in order to fix a system we
have to specify the scale in which we are considering it) is at the
heart of the notion of effective field theory, where the scale (or
parameter) is governed by Renormalization Group flows \citep{6,7,8,9}.
Notice, in turn, that if a system emerges from some other, then the
second one should be more fundamental, at least in the scale (or parameter)
in which the emergence phenomenon is observed. This also puts Emergence
Theory in the framework of searching for the fundamental theory of
Physics (e.g Quantum Gravity), whose systems should be the minimal
systems relative to the emergence relation \citep{3,4}. The main
problem in this setting is then the existence problem for the minimum.
A very related question is the general existence problem: \emph{given
two systems, is there some emergence relation between them?}

One can work on the existence problem is different levels of depth.
Indeed, since the systems is question are parameterized one can ask
if there exists a correspondence between them in \emph{some scales}
or in \emph{all scales}. Obviously, by requiring a complete correspondence
between them is much more strong than requiring a partial one. On
the other hand, in order to attack the existence problem we also have
to specify which kind of emergence relation we are looking for. Again,
is it a full correspondence, in the sense that the emergent theory
can be fully recovered from the fundamental one, or is it only a 
partial correspondence, through which only certain aspects can be
recovered? Thus, we can say that we have the following four versions
of the existence problem for emergence phenomena.

\begin{table}[H]
\begin{centering}
\begin{tabular}{|c|c|c|c|c|}
\cline{2-5} 
\multicolumn{1}{c|}{} & \emph{weak} & \emph{weak-scale} & \emph{weak-relation} & \emph{strong}\tabularnewline
\hline 
\emph{relation} & partial & full & partial & full\tabularnewline
\hline 
\emph{scales} & some & some & all & all\tabularnewline
\hline 
\end{tabular}
\par\end{centering}
\caption{Types of Emergence}
\end{table}

In Physics one usually works on finding weak emergence phenomena.
Indeed, one typically shows that certain properties of a system can
be described by some other system at some limit, corresponding to
a certain regime of the parameter space. These emergence phenomena
are strongly related with other kind of relation: the \emph{physical
duality}, where two different systems reveal the same physical properties.
One typically builds emergence from duality. For example, AdS/CFT
duality plays an important role in describing spacetime geometry (curvature)
from mechanic statistical information (entanglement entropy) of dual
strongly coupled systems \citep{10,11,12,13,14,15}. 

There are also some interesting examples of weak-scale emergence relations,
following again from some duality. These typically occur when the
action functional of two Lagrangian field theories are equal at some
limit. The basic example is gravity emerging from noncommutativity
following from the duality between commutative and noncommutative
gauge theories established by the Seiberg-Witten maps \citep{16}.
Quickly, the idea was to consider a gauge theory $S[A]$ and modify
it into two different ways: 
\begin{enumerate}
\item by considering $S[A]$ coupled to some background field $\chi$, i.e,
$S_{\chi}[A;\chi]$;
\item by using the Seiberg-Witten map to get its noncommutative analogue
$S_{\theta}[\hat{A};\theta]$.
\end{enumerate}
$\quad\;\,$Both new theories can be regarded as parameterized theories:
the parameter (or scale) of the first one is the background field
$\chi$, while that of the second one is the noncommutative parameter
$\theta^{\mu\nu}$. By construction, the noncommutative theory $S_{\theta}[A;\theta]$
can be expanded in a power series on the noncommutative parameter,
and we can also expand the other theory $S_{\chi}[A;\chi]$ on the
background field, i.e, one can write 
\[
S_{\chi}[A;\chi]=\sum_{i=0}^{\infty}S_{i}[A;\chi^{i}]=\lim_{n\rightarrow\infty}S_{(n)}[A;\chi]\quad\text{and}\quad S_{\theta}[\hat{A};\theta]=\sum_{i=0}^{\infty}S_{i}[A;\theta^{i}]=\lim_{n\rightarrow\infty}S_{(n)}[A;\theta],
\]
where $S_{(n)}[A;\chi]=\sum_{i=0}^{n}S_{i}[A;\chi^{i}]$ and $S_{(n)}[A;\theta]=\sum_{i=0}^{n}S_{i}[A;\theta^{i}]$
are partial sums. One then try to find solutions for the following
question:
\begin{question}
\label{question_int}Given a gauge theory $S[A]$, is there a background
version  $S_{\chi}[A;\chi]$ of it and a number $n$ such that for
every given value $\theta^{\mu\nu}$ of the noncommutative parameter
there exists a value of the background field $\chi(\theta)$, possibly
depending on $\theta^{\mu\nu}$, such that for every gauge field $A$
we have $S_{(n)}[A;\chi(\theta)]=S_{(n)}[A;\theta]$?
\end{question}
Notice that if rephrased in terms of parameterized theories, the question
above is precisely about the existence of a weak-scale emergence between
$S_{\chi}$ and $S_{\theta}$, at least up to order $n$. This can
also be interpreted by saying that, in the context of the gauge theory
$S[A]$, the background fields $\chi$ emerges in some regime from
the noncommutativity of the spacetime coordinates. Since the noncommutative
parameter $\theta^{\mu\nu}$ depends on two spacetime indexes, it
is suggestive to consider background fields of the same type, i.e, 
$\chi^{\mu\nu}$. In this case, there is a natural choice: metric
tensors $g^{\mu\nu}$. Thus, in this setup, the previous question
is about proving that in the given gauge context, gravity emerges
from noncommutativity at least up to a perturbation of order $n$.
This was proved to be true for many classes of gauge theories and
for many values of $n$ \citep{17,18,19,20,21}. On the other hand,
this naturally leads to other two questions:
\begin{enumerate}
\item Can we find some emergence relation between gravity and noncommutativity
in the nonperturbative setting? In other words, can we extend the
weak-scale emergence relation above to a strong one? 
\item Is it possible to generalize the construction of the cited works to
other kind of background fields? In other words, is it possible to
use the same idea in order to show that different fields emerge from
spacetime noncommutativity?
\end{enumerate}
$\quad\;\,$The first of these questions is about finding a strong
emergence phenomena and it has a positive answer in some cases \citep{22,23,24,25}.
The second one, in turn, is about working on finding systematic and
general conditions ensuring the existence (or nonexistence) of emergence
phenomena. At least to the authors knowledge, there are no such studies,
specially focused on the strong emergence between field theories.
It is precisely this point that is the focus of the present work.
Indeed we will:
\begin{enumerate}
\item based on Question \ref{question_int}, propose an axiomatization for
the notion of \emph{strong emergence} between field theories;
\item establish sufficient conditions ensuring that a given Lagrangian field
theory emerges from each theory belonging to a certain class of theories;
\item show that the given sufficient conditions are not necessary conditions.
\end{enumerate}
$\quad\;\,$In the remainder of this Introduction, let us be a bit
more explicit about our aim. For us, a \emph{field theory} over a
oriented $n$-dimensional manifold $M$ (regarded as the spacetime)
is defined by an action functional $S[\varphi]=\int_{M}\mathcal{L}(x,\varphi,\partial\varphi)d^{n}x$,
where $\mathcal{L}$ is the Lagrangian density and $\varphi$ is some
generic field (section of some real or complex vector bundle $E\rightarrow M$,
the \emph{field bundle}). A \emph{parameterized field theory} consists
of another bundle $P\rightarrow M$ (the \emph{parameter bundle}),
a subset $\operatorname{Par}(P)\subset\Gamma(P)$ of global sections
(the \emph{parameters})\emph{ }and a collection $S_{\varepsilon}[\varphi]=\int_{M}\mathcal{L}_{\varepsilon}(x,\varphi,\partial\varphi)d^{n}x$
of field theories, one for each parameter $\varepsilon\in\operatorname{Par}(P)$.
A more suggestive notation should be $S[\varphi;\varepsilon]$ and
$\mathcal{L}(x,\varphi,\partial\varphi;\varepsilon)$. So, e.g, for
the trivial parameter bundle $P\simeq M\times\mathbb{K}$ we have
$\Gamma(P)\simeq C^{\infty}(M;\mathbb{K})$ and in this case we say
that we have \emph{scalar parameters.} If we consider only scalar
parameters which are constant functions, then a parameterized theory
becomes the same thing as a 1-parameter family of field theories.
Here, and throughout the paper, $\mathbb{K}=\mathbb{R}$ or $\mathbb{K}=\mathbb{C}$
depending on whether the field bundle in consideration is real or
complex.

We will think of a parameter $\varepsilon$ as some kind of ``physical
scale'', so that for two given parameters $\varepsilon$ and $\varepsilon'$,
we regard $S[\varphi;\varepsilon]$ and $S[\varphi;\varepsilon']$
as \emph{the same theory} in two different physical scales. Notice
that if $P$ has rank $l$, then we can locally write $\varepsilon=\sum\varepsilon^{i}e_{i}$,
with $i=1,...,l$, where $e_{i}$ is a local basis for $\Gamma(P)$.
Thus, locally each physical scale is completely determined by $l$
scalar parameters $\varepsilon^{i}$ which are the fundamental ones.
In terms of these definitions, Question \ref{question_int} has a
natural generalization:
\begin{question}
\label{problem_1}Let $S_{1}[\varphi;\varepsilon]$ and $S_{2}[\psi;\delta]$
be two parameterized theories defined on the same spacetime $M$,
but possibly with different field bundles $E_{1}$ and $E_{2}$, and
different parameter bundles $P_{1}$ and $P_{2}$. Arbitrarily giving
a field $\varphi\in\Gamma(E_{1})$ and a parameter $\varepsilon\in\operatorname{Par}_{1}(P_{1})$,
can we find some field $\psi(\varphi)\in\Gamma(E_{2})$ and some parameter
$\delta(\varepsilon)\in\operatorname{Par}_{2}(P_{2})$ such that $S_{1}[\varphi;\varepsilon]=S_{2}[\psi(\varphi);\delta(\varepsilon)]$?
In more concise terms, are there functions $F:\operatorname{Par}_{1}(P_{1})\rightarrow\operatorname{Par}_{2}(P_{2})$
and $G:\Gamma(E_{1})\rightarrow\Gamma(E_{2})$ such that $S_{1}[\varphi;\varepsilon]=S_{2}[G(\varphi);F(\varepsilon)]$?
\end{question}
We say that the theory $S_{1}[\varphi;\varepsilon]$ \emph{emerges}
\emph{from the theory }$S_{2}[\psi;\delta]$ if the problem above
has a positive solution, i.e, if we can fully describe $S_{1}$ in
terms of $S_{2}$. Notice, however, that as stated the emergence problem
is fairly general. Indeed, if $P_{1}$ and $P_{2}$ have different
ranks, then, by the previous discussion, this means that the parameterized
theories $S_{1}$ and $S_{2}$ have a different number of fundamental
scales, so that we should not expect an emergence relation between
them. This leads us to think of considering only the case in which
$P_{1}=P_{2}$. However, we could also consider the situations in
which $P_{1}\neq P_{2}$, but $P_{2}=f(P_{1})$ is some nice function
of $P_{1}$, e.g, $P_{2}=P_{1}\times P_{1}\times...\times P_{1}$.
In these cases the fundamental scales remain only those of $P_{1}$,
since from them we can generate those in the product. Throughout this
paper we will also work with different theories defined on the same
fields, i.e, $E_{1}=E_{2}$. This will allow us to search for emergence
relations in which $G$ is the identity map $G(\varphi)=\varphi$. 

Hence, after these hypotheses, we can rewrite our main problem, whose
affirmative solutions axiomatize the notion of strong emergence we
are searching for:
\begin{question}
\label{problem_2}Let $S_{1}[\varphi;\varepsilon]$ and $S_{2}[\varphi;\delta]$
be two parametrized theories defined on the same spacetime $M$, on
the same field bundle $E$ and on the parameter bundles $P_{1}$ and
$P_{2}=f(P_{1})$, respectively. Does there exists some map $F:\operatorname{Par}_{1}(P_{1})\rightarrow\operatorname{Par}_{2}(f(P_{1}))$
such that $S_{1}[\varphi;\varepsilon]=S_{2}[\varphi,F(\varepsilon)]$?
\end{question}
Our plan is to show that the problem in Question \ref{problem_2}
has an affirmative solution in some interesting cases. We begin by
noticing that when working in a spacetime without boundary, after
integration by parts and using Stoke's theorem, many of the typical
field theories can be stated, at least locally, in the form $\mathcal{L}(x,\varphi,\partial\varphi)=\langle\varphi,D\varphi\rangle$,
where $\langle\varphi,\varphi'\rangle$ is a nondegenerate pairing
on the space of fields $\Gamma(E)$ and $D:\Gamma(E)\rightarrow\Gamma(E)$
is a differential operator of degree $d$, which means that it can
be locally written as $D\varphi(x)=\sum_{\vert\alpha\vert\leq d}a_{\alpha}(x)\partial^{\alpha}\varphi$,
where $\alpha=(\alpha_{1},...,\alpha_{r})$ is some mult-index, $\vert\alpha\vert=\alpha_{1}+...+\alpha_{r}$
is its degree and $\partial^{\alpha}=\partial_{1}^{\alpha_{1}}\circ...\circ\partial_{r}^{\alpha_{r}}$,
with $\partial_{i}^{l}=\partial^{l}/\partial^{l}x_{i}$. Let $\operatorname{Diff}^{d}(E;E)$
denote the space of differential operators of degree $l$. This is
the case, e.g, of $\varphi^{3}$ and $\varphi^{4}$ scalar field theories,
the standard spinorial field theories as well as Yang-Mills theories.
More generally, recall that the first step in building the Feynman
rules of a field theory is to find the (kinematic par of the) operator
$D$ and take its ``propagator'' .

Typically, the pairing $\langle\varphi,\varphi'\rangle$ is symmetric
(resp. skew-symmetric) and the operator $D$ is formally self-adjoint
(resp. formally anti-self-adjoint) relative to that pairing. Furthermore,
$\langle\varphi,\varphi'\rangle$ is usually a $L^{2}$-pairing induced
by a semi-Riemannian metric $g$ on the field bundle $E$ and/or on
the spacetime $M$, while $D$ is usually a generalized Laplacian
or a Dirac-type operator relative to $g$ \citep{costello2011renormalization}.
For example, this holds for the concrete field theories (scalar, spinorial
and Yang-Mills) above. The skew-symmetric case generally arises in
gauge theories (BV-BRST quantization) after introducing the Faddeev-Popov
ghosts/anti-ghosts and it depends on the grading introduced by the
ghost number \citep{costello2011renormalization}.

Another remark, still having in mind the concrete situations above,
is that if the metric $g$ inducing the pairing $\langle\varphi,\varphi'\rangle$
is actually Riemannian (which means that the gravitational background
is Euclidean), then $\langle\varphi,\varphi'\rangle$ becomes a genuine
$L^{2}$-inner product and $D$ is elliptic and extends to a bounded
self ajoint operator between Sobolev spaces \citep{donaldson1990geometry,freed2012instantons}.
Working with elliptic operators is very useful, since they always
admits parametrices (which in this Euclidean cases are the propagators)
and for generalized Laplacians the heat kernel not only exists, but
also has a well-known asymptotic behavior \citep{heat_kernel_generalized_laplacian},
which is very nice in the Dirac-type case \citep{berline2003heat}.

From the discussion above, it is natural to focus on parameterized
theories whose parameterized Lagrangian densities are of the form
$\mathcal{L}(x,\varphi,\partial\varphi;\varepsilon)=\langle\varphi,D_{\varepsilon}\varphi\rangle$,
i.e, are determined by a single nondegenerate pairing $\langle\varphi,\varphi'\rangle$
in $\Gamma(E)$, fixed a priori by the nature of $M$ and $E$, and
by a family of differential operators $D_{\varepsilon}\in\operatorname{Diff}(E)$,
one for each parameter $\varepsilon\in\operatorname{Par}(P)$, where
$\operatorname{Diff}(E)=\bigoplus_{d}\operatorname{Diff}^{d}(E;E)$
denotes the space of all differential operators in $E$. We will assume
that the pairing extends to some space $\mathcal{H}(E)$, containing
$\Gamma(E)$ as a dense subspace, such that the corresponding map
$\ll\varphi,\varphi'\gg=\int_{M}\langle\varphi,\varphi'\rangle d^{n}x$
turns $\mathcal{H}(E)$ into a Hilbert space (typically a Sobolev
space). We will further assume that each action functional $S_{\varepsilon}:\Gamma(E)\rightarrow\mathbb{R}$
is continuous and that each operator $D_{\varepsilon}$ extends to
a bounded operator $\hat{D}_{\varepsilon}:\mathcal{H}(E)\rightarrow\mathcal{H}(E)$.
The denseness of $\Gamma(E)$ in $\mathcal{H}(E)$ then implies that
each $S_{\varepsilon}$ has a continuous extension $\hat{S}_{\varepsilon}\mathcal{H}(E)\rightarrow\mathbb{R}$,
given by $\hat{S}[\varphi;\varepsilon]=\int_{M}\langle\varphi,\hat{D}_{\varepsilon}\varphi\rangle d^{n}x$.
We will call these theories \emph{differential parameterized theories
}(DPT)\emph{ defined by the differential operators $D_{\varepsilon}$}.

Notice that in building the ``propagator'' of $D$ we are actually
finding some kind of ``quasi-inverse'' $D^{-1}$. For elliptic operators
(where parametrices are propagators) the situation becomes more clear.
Indeed, a parametrix is an inverse (up to compact operators) for the
extended bounded linear map $\hat{D}:\mathcal{H}(E)\rightarrow\mathcal{H}(E)$.
It would be very useful if the quasi-inverse $D^{-1}$ could exist
as a differential operator, i.e, if $D^{-1}=\hat{Q}$ for some differential
operator $Q$. However, this is not the case. But $D^{-1}$ may exist
as a more generalized class of objects: \emph{pseudo-differential
operators} \citep{ghaemi2016study}. Since every differential operator
is a pseudo-differential operator, we see that the process of building
propagators (at least for elliptic operators in an Euclidean background)
is more workable in the language of pseudo-differential operators.

Even so, as we will see, for many purposes it is better to work with
theories defined in a suitable extension $\operatorname{Nice}(E)\supset\operatorname{Diff}(E)$
of the space of differential operators, consisting of certain \emph{nice
objects} $\Psi_{\varepsilon}$ (e.g, pseudo-differential operators),
which can also be regarded as bounded operators $\tilde{\Psi}_{\varepsilon}:\mathcal{H}(E)\rightarrow\mathcal{H}(E)$.
For example, some useful operations are not defined for arbitrary
differential operators (like taking the inverse), but they are in
some better-behaved context. Let us call theories defined in $\operatorname{Nice}(E)$
as \emph{nice parametrized theory (NPT)}. We will also need some additional
structure on the parameter space $\operatorname{Par}(P)$. Indeed,
we will need to sum distinct parameters ($\varepsilon+\varepsilon'$,
with $\varepsilon\neq\varepsilon'$) and multiply two arbitrary parameters
($\varepsilon*\varepsilon'$, with possibly $\varepsilon=\varepsilon'$).
This fits into a structure that we call \emph{nowhere vanishing algebra}.
We also require the existence of square roots in $\operatorname{Par}(P)$,
meaning that for every $\varepsilon$ there exists certain $\sqrt{\varepsilon}$
such that $\sqrt{\varepsilon}*\sqrt{\varepsilon}=\varepsilon$. If
a NPT $S_{\varepsilon}$ is such that $\operatorname{Nice}(E)$ is
a $\mathbb{K}$-algebra with the composition operation, $\operatorname{Par}(P)$
is a nowhere vanishing algebra and the rule $\varepsilon\mapsto\Psi_{\varepsilon}$
preserves sum and multiplication, then we say that $S_{\varepsilon}$
is \emph{partially homomorphic on parameters}.

Finally, we will need to multiply operators $\operatorname{Nice}(E)$
by families of parameters. This means we will need actions $\cdot^{\ell}:\operatorname{Par}(P)^{\ell}\times\operatorname{Nice}(E)\rightarrow\operatorname{Nice}(E)$.
These are bilinear maps which are compatible with the composition
of nice operators, in the sense that
\begin{equation}
(\varepsilon(\ell)\cdot^{\ell}\Psi)\circ\Psi'=\varepsilon(\ell)\cdot^{\ell}(\Psi\circ\Psi').\label{compatibility_action}
\end{equation}
Here, given a non-negative integer $\ell\geq0$, $\operatorname{Par}(P)^{\ell}=\operatorname{Par}(P)\times...\times\operatorname{Par}(P)$
denotes the $\ell$-power of the parameter set. For short, if $\ell>0$
we will denote an element of the product by $\varepsilon(\ell)$,
so that $\varepsilon(\ell)=(\varepsilon_{1},...,\varepsilon_{\ell})$,
with $\varepsilon_{i}\in\operatorname{Par}(P)$. We will also use
the convention that $\operatorname{Par}(P)^{0}$ is a singleton, whose
element we denote by $\varepsilon(0)$. For a fixed $\Psi$ we get
a map 
\begin{equation}
r_{\Psi}^{\ell}:\operatorname{Par}(P)^{\ell}\rightarrow\operatorname{Nice}(E)\quad\text{given by}\quad r_{\Psi}^{\ell}(\varepsilon(\ell))=\varepsilon(\ell)\cdot^{\ell}\Psi,\label{right_multiplication}
\end{equation}
i.e, the right-hand side multiplication by $\Psi$. We will demand
that for every $\ell$ the map $r_{I}^{\ell}$ is injective if we
fix the identity operator $\Psi=I$. Thus, for every $\varepsilon(\ell),\delta(\ell)$
we have $\varepsilon(\ell)\cdot^{\ell}I=\varepsilon(\ell)\cdot^{\ell}I$
iff $\varepsilon(\ell)=\varepsilon(\ell)$, i.e, iff $\varepsilon_{i}=\delta_{i}$,
with $i=1,...,\ell$. A NPT $S_{\varepsilon}$ whose operator algebra
$\operatorname{Nice}(E)$ is endowed with the actions $\cdot^{\ell}$
is called a \emph{NPT with action} \emph{by} \emph{parameters of degree
$\ell$}. 

The fundamental property which will be required on the actions $\cdot^{\ell}$
is that they must allow a nice functional calculus in $\operatorname{Nice}(E)$.
In order to be more precise, let $\operatorname{Map}(\operatorname{Par}(P)^{\ell};\mathbb{K})$
denote the set of all functions $f:\operatorname{Par}(P)^{\ell}\rightarrow\mathbb{K}$.
Thus, if $\ell=0$, functions $f:\operatorname{Par}(P)^{0}\rightarrow\mathbb{K}$
are identified with the number $f(\varepsilon(0))$, so that $\operatorname{Map}(\operatorname{Par}(P)^{0};\mathbb{K})\simeq\mathbb{K}$.
Notice that the scalar multiplication $\cdot:\mathbb{K}\times\operatorname{Nice}(E)\rightarrow\operatorname{Nice}(E)$,
which exists since $\operatorname{Nice}(E)$ is a $\mathbb{K}$-algebra,
induces maps 
\[
\cdot\Psi:\operatorname{Map}(\operatorname{Par}(P)^{\ell};\mathbb{K})\rightarrow\operatorname{Map}(\operatorname{Par}(P)^{\ell};\operatorname{Nice}(E))\quad\text{given by}\quad[f\cdot\Psi](\varepsilon(\ell))=f(\varepsilon(\ell))\Psi.
\]
Thus, in a similar way, the actions $\cdot^{\ell}$ above induce 
\[
R_{\Psi}^{\ell}:\operatorname{Map}(\operatorname{Par}(P)^{\ell};\operatorname{Par}(P)^{\ell})\rightarrow\operatorname{Map}(\operatorname{Par}(P)^{\ell};\operatorname{Nice}(E))\quad\text{given by}\quad[R_{\Psi}^{\ell}F](\varepsilon(\ell))=F(\varepsilon(\ell))\cdot^{\ell}\Psi
\]

A \emph{functional calculus in $\operatorname{Nice}(E)$} \emph{compatible
with the action $\cdot^{\ell}$ }consists of a set $C_{\ell}(P;\mathbb{K})$
of functions $f:\operatorname{Par}(P)^{\ell}\rightarrow\mathbb{K}$
for each $\ell\geq0$, together with a map $\Psi_{-}^{\ell}:C_{\ell}(P;\mathbb{K})\rightarrow\operatorname{Nice}(E)$
assigning to each function $f$ a corresponding operator $\Psi_{f}^{\ell}$
with the property that $\Psi_{f}^{\ell}\circ(f\cdot\Psi)=id\cdot^{\ell}\Psi$
for every nice operator $\Psi$\footnote{This condition could be weakened by requiring $\Psi_{f}^{\ell}\circ(f\cdot\Psi)=id\cdot^{\ell}\Psi$
only for $\Psi$ belonging to a subalgebra $\operatorname{Nice}_{0}(E)\subset\operatorname{Nice}(E)$.
This would produce a slight generalization of some steps in the proof
of our main theorem. However, for simplicity we will assume the existence
of a functional calculus in the whole algebra $\operatorname{Nice}(E)$.}.\emph{ }Explicitly, this means that 
\begin{equation}
\Psi_{f}^{\ell}\circ[f(\varepsilon(\ell))\Psi]=\varepsilon(\ell)\cdot^{\ell}\Psi.\label{functional_calculus}
\end{equation}
Let $S_{\varepsilon(\ell)}$ be a NPT with action by parameters of
degree $\ell$. If its operator algebra is endowed with a functional
calculus compatible with $\cdot^{\ell}$, say by a set of functions
$C_{\ell}(P;\mathbb{K})$, we will say that $S_{\varepsilon(\ell)}$
is a \emph{NPT with functional calculus defined on $C(\operatorname{Par}(P)^{\ell};\mathbb{K})$}.
The following example reveals some interest in nowhere vanishing functions.
\begin{example}
\label{nowhere_vanishing_calculus}Every NPT $S_{\varepsilon(\ell)}$
with action $\cdot^{\ell}$ admits a unique structure of NPT with
functional calculus defined on the set $\operatorname{Map}_{\neq0}(\operatorname{Par}(P)^{\ell};\mathbb{K})$
of nowhere vanishing functions. In order to refer to this canonical
structure we will say simply that $S_{\varepsilon(\ell)}$ is a \emph{NPT
with nowhere vanishing functional calculus}. Assume that it exists.
Then $\Psi_{f}^{\ell}\circ[f(\varepsilon(\ell))\Psi]=\varepsilon(\ell)\cdot^{\ell}\Psi$
for every $f$, $\Psi$ and $\varepsilon(\ell)$, so that $\Psi_{f}^{\ell}\circ\Psi=(\varepsilon(\ell)\cdot^{\ell}\Psi)/f(\varepsilon(\ell))$.
In particular, for $\Psi=I$, we get $\Psi_{f}^{\ell}=(\varepsilon(\ell)\cdot^{\ell}I)/f(\varepsilon(\ell))$.
In order to prove uniqueness, define $\Psi_{f}^{\ell}=(\varepsilon(\ell)\cdot^{\ell}I)/f(\varepsilon(\ell))$,
so that 
\[
\Psi_{f}^{\ell}\circ[f(\varepsilon(\ell))\Psi]=(\varepsilon(\ell)\cdot^{\ell}I)\circ\Psi=\varepsilon(\ell)\cdot^{\ell}(I\circ\Psi)=\varepsilon(\ell)\cdot^{\ell}\Psi,
\]
where in the last step we used the compatibility between $\cdot^{\ell}$
and $\circ$, as described in (\ref{compatibility_action}).
\end{example}
Given integers $l,\ell\geq0$ and $r>0$, let $p_{\ell}^{l}[x_{1},...,x_{r}]=\sum_{\vert\alpha\vert\leq l}f_{\alpha}\cdot x^{\alpha}$
be some multivariable polynomial of degree $l$ whose coefficients
are functions $f_{\alpha}:\operatorname{Par}(P)^{\ell}\rightarrow\mathbb{K}$.
Thus, e.g, the polynomials $p_{0}^{l}[x_{1},...,x_{r}]$ are precisely
the classical polynomials with coefficients in $\mathbb{K}$. If $\Psi_{1},...,\Psi_{r}\in\operatorname{Nice}(E)$
are fixed operators and $p_{\ell}^{l}[x_{1},...,x_{r}]$ is a polynomial
as above, then by means of replacing the formal variables $x_{i}$
with the operators $\Psi_{i}$ we get another operator $p_{\ell}^{l}[\Psi_{1},...,\Psi_{r}]=\sum_{\vert\alpha\vert\leq l}f_{\alpha}\cdot\Psi^{\alpha}$,
now depending on $\ell$ parameters. Indeed, for every $\varepsilon(\ell)\in\operatorname{Par}(P)$
we get $p_{\ell}^{l}[\Psi_{1},...,\Psi_{r}]_{\varepsilon(\ell)}\in\operatorname{Nice}(E)$
such that $p^{l}[\Psi_{1},...,\Psi_{r}]_{\varepsilon(\ell)}=\sum_{\vert\alpha\vert\leq l}f_{\alpha}(\varepsilon(\ell))\Psi^{\alpha}$.
Let us say that a NPT is a \emph{polynomial theory of degree $(l,\ell)$
in $r$ variables} if its parameterized operator is of the form $p_{\ell}^{l}[\Psi_{1},...,\Psi_{r}]$
for certain polynomial $p_{\ell}^{l}[x_{1},...,x_{r}]$ and certain
operators $\Psi_{1},...,\Psi_{r}$. 

We can now state the main result of this paper. It says that typical
parameterized theories $S_{\varepsilon}$ emerges from any suitably
multivariate polynomial theory.$\underset{\underset{\;}{\;}}{\;}$

\noindent \textbf{Main Theorem. }\emph{Let $S_{1,\varepsilon(\ell)}$
be a NPT which is partially homomorphic on parameters, whose parameter
algebra $\operatorname{Par}(P)^{\ell}$ has square roots. Then $S_{1,\varepsilon(\ell)}$
emerges from any NPT $S_{2,\delta(\ell')}$, with functional calculus
$C_{\ell'}(P;\mathbb{K})$, which is a polynomial theory of degree
$(l,\ell')$ in $r$ variables, whose defining polynomial $p_{\ell'}^{l}[\Psi{}_{1},...,\Psi_{r}]$
has coefficients given by functions $f_{\alpha}$ belonging to $C_{\ell'}(P;\mathbb{K})$
and whose operators $\Psi_{1},...,\Psi_{r}$ are right-invertible.}$\underset{\underset{\;}{\;}}{\;}$

The proof will be done in several steps.
\begin{enumerate}
\item We first prove the case $p_{\ell'}^{1}[\Psi]_{\delta(\ell')}=g(\delta(\ell'))\Psi$.
\item Then we obtain an additivity result for the emergence problem.
\item Next we obtain a multiplicativity result for the emergence problem.
\item We use the previous steps and some additional hypotheses on the functional
calculus in order to prove the case $p_{\ell'}^{l}[\Psi]_{\delta_{1},...,\delta_{l}}$
of an arbitrary univariate polynomial.
\item Then we prove a recurrence result for the emergence problem.
\item Finally use all the stebs above to prove the general multivariate
case $p_{\ell'}^{l}[\Psi_{1},...,\Psi_{r}]$.
\end{enumerate}
$\quad\;\,$The paper is organized as follows. In Section \ref{sec_notations}
the previous discussion is reviewed, now in more precise and rigorous
terms. In Sections \ref{steps_1_2}-\ref{step_5} we prove the first
five steps described above, while in Section \ref{sec_theorem} the
main theorem is restated in a more concise form and then proved. Some
concrete examples of our methods are given in Section \ref{sec_examples},
where we also try to emphasize the true range of our results.

\section{\label{sec_notations} Definitions, Notations and Remarks}

$\quad\;\,$Let us begin by recalling (and presenting in more details)
some definitions briefly presented at the introduction. A \emph{classical
background for doing emergence theory, }denoted by\emph{ $\mathcal{CB}$,
}is given by the following data:
\begin{itemize}
\item a compact oriented\footnote{\label{nao_precisa_compact}Most of the results hold without compactness
and orientability hypotheses. Instead, we need only assume integrability
conditions on global sections and consider Lagrangians as taking values
on general densities.} smooth manifold $M$;
\item a real/complex vector bundle $E\rightarrow M$ (the field bundle);
\item a vector space $\mathcal{H}(E)$ containing $\Gamma(E)$ as a dense
subspace;
\item a pairing $\langle\varphi,\varphi'\rangle$ in $\mathcal{H}(E)$ such
that $\ll\varphi,\varphi'\gg=\int\langle\varphi,\varphi'\rangle d^{n}x$
turns it into a Hilbert space;
\item a subset $\operatorname{Diff}_{0}(E)\subset\operatorname{Diff}(E)$;
\item an injective linear map $\hat{\cdot}:\operatorname{Diff}(E)\rightarrow B(\mathcal{H}(E))$;
\item a parameter bundle $P\rightarrow M$ and a set of parameters $\operatorname{Par}(P)\subset\Gamma(P).$
\end{itemize}
A \emph{differential parameterized theory (DPT)} in that classical
background \emph{$\mathcal{CB}$} is given by 
\begin{itemize}
\item an integer $\ell\geq0$, called the \emph{parameter degree};
\item for each parameter $\varepsilon(\ell)\in\operatorname{Par}(P)^{\ell}$
a differential operator $D_{\varepsilon(\ell)}\in\operatorname{Diff}_{0}(E)$.
\end{itemize}
$\quad\;\,$The\emph{ action functional} and the \emph{extended action
functional} are then defined by $S_{\varepsilon(\ell)}[\varphi;\varepsilon(\ell)]=\int\langle\varphi,D_{\varepsilon(\ell)}\varphi\rangle d^{n}x$
and $\hat{S}_{\varepsilon(\ell)}[\varphi]=\int\langle\varphi,\hat{D}_{\varepsilon(\ell)}\varphi\rangle d^{n}x.$
The linearity of the map $\hat{\cdot}:\operatorname{Diff}(E)\rightarrow B(\mathcal{H}(E))$
means that $\widehat{D_{1}+D_{2}}=\hat{D}_{1}+\hat{D}_{2}$ and that
$\widehat{cD}=c\hat{D}$, where $c\in\mathbb{C}$ is a scalar. In
the following we need to embed the space of differential operators
into a more interesting space than $B(\mathcal{H}(E))$: the space
$RB(\mathcal{H}(E))$ of bounded operators $T$ which have a bounded
right-inverse, i.e, which admit a retraction $R$ such that $T\circ R=I$.
The problem is that given $\operatorname{Diff}_{0}(E)\subset\operatorname{Diff}(E)$
in general there is no injective linear map $\hat{\cdot}:\operatorname{Diff}(E)\rightarrow B(\mathcal{H}(E))$
factoring through $RB(\mathcal{H}(E))$, as pictured below.$$
\xymatrix{\operatorname{Diff}(E) \ar[r]^{\hat{\cdot}} & B(\mathcal{H}(E)) \\
 \ar@{^(->}[u]\operatorname{Diff}_{0}(E) \ar@{-->}[r] & RB(\mathcal{H}(E))  \ar@{^(->}[u]}
$$

The solution is to modify the strategy: instead of considering parameterized
theories defined by linear maps $D_{\varepsilon}:\Gamma(E)\rightarrow\Gamma(E)$
which belongs to a \emph{rigid space} $\operatorname{Diff}_{0}(E)\subset\operatorname{Diff}(E)$
of \emph{well-behaved objects}, the idea is to consider theories defined
by maps in a \emph{well-behaved space} $\operatorname{Nice}(E)\supset\operatorname{Diff}(E)$
of \emph{nice objects} $\Psi_{\varepsilon}$. This leads us to a define
a \emph{nice background for doing emergence theory}, denoted by $\mathcal{NC}$,
as being given by:
\begin{itemize}
\item a vector bundle $E\rightarrow M$ (the field bundle);
\item a vector space $\mathcal{H}(E)$ containing $\Gamma(E)$ as a dense
subspace;
\item a pairing $\langle\varphi,\varphi'\rangle$ in $\mathcal{H}(E)$ such
that $\ll\varphi,\varphi'\gg=\int\langle\varphi,\varphi'\rangle d^{n}x$
turns it into a Hilbert space;
\item an algebra $\operatorname{Nice}(E)\subset\operatorname{End}(\Gamma(E))$;
\item an injective algebra homomorphism $\tilde{\cdot}:\operatorname{Nice}(E)\rightarrow B(\mathcal{H}(E))$;
\item a parameter bundle $P\rightarrow M$ and a set of parameters $\operatorname{Par}(P)\subset\Gamma(P)$.
\end{itemize}
A \emph{nice parameterized theory }in a nice background $\mathcal{NB}$
consists of
\begin{itemize}
\item an integer $\ell\geq0$, called the \emph{parameter degree};
\item for each parameter $\varepsilon(\ell)\in\operatorname{Par}(P)^{\ell}$
an element $\Psi_{\varepsilon(\ell)}\in\operatorname{Nice}(E)$.
\end{itemize}
$\quad\;\,$The\emph{ action functional} and the \emph{extended action
functional }for a NPT are defined analogously by $S_{\varepsilon(\ell)}[\varphi;\varepsilon(\ell)]=\int\langle\varphi,\Psi_{\varepsilon(\ell)}\varphi\rangle d^{n}x$
and $\tilde{S}_{\varepsilon(\ell)}[\varphi]=\int\langle\varphi,\tilde{\Psi}_{\varepsilon(\ell)}\varphi\rangle d^{n}x.$
The \emph{right-invertible operators }in a nice background are those
$\Psi$ such that $\tilde{\Psi}\in RB(\mathcal{H}(E))$, i.e, such
that there exists $R_{\tilde{\Psi}}$ with $\tilde{\Psi}\circ R_{\tilde{\Psi}}=I$.
In this case, there also exists $R_{\Psi}\in\operatorname{Nice}(E)$
such that $\widetilde{R_{\Psi}}=R_{\tilde{\Psi}}$.
\begin{example}
Every classical background in which $\operatorname{Diff}_{0}(E)$
is an algebra can be regarded as a nice background by taking $\operatorname{Nice}(E)=\operatorname{Diff}_{0}(E)$.
However, they generally do not have right-inverses.
\begin{example}
The main setup in building nice backgrounds which admit right-invertible
operators is to take $M$ a compact Riemannian manifold, $E=M\times\mathbb{C}$
a scalar field bundle, $\mathcal{H}(E)=\bigoplus_{k}H^{k}(M)$ a sum
of Sobolev spaces, $\operatorname{Nice}(E)$ the algebra of scalar
pseudo-differential operators (as suggested by the notation $\Psi$)
and the map $\tilde{\cdot}$ as the canonical extensions of a pseudo-differential
operator as bounded operators between Sobolev spaces. We then look
for ellipticity conditions to ensure right-inverses \citep{right_inverse_const_1,right_inverse_const_2,right_inv_1,right_inv_1_2,right_inv_2,right_inv_3,right_inv_4}.
\end{example}
\end{example}
\begin{rem}
\label{remark_notations}Throughout this paper we will consider NPT
defined in a fixed nice background. Thus, instead of saying \emph{``let
$S_{1,\varepsilon(\ell)}$ and $S_{2,\delta(\ell')}$ be two NPT,
with parameter degrees $\ell$ and $\ell$', defined in a nice background
$\mathcal{NB}$}'' we will say \emph{``let $S_{1,\varepsilon(\ell)}$
and $S_{2,\delta(\ell')}$ be two NPT} w\emph{ith parameter degrees
$\ell$ and $\ell$}'', leaving the fixed nice background implicit.
Actually, many results will be independent of the parameter degree
of the emerging theory. In order to emphasize this fact we will omit
the parameter degree when its value does not matter. Thus, in such
situations we will simply say ``\emph{let $S_{1,\varepsilon}$ and
$S_{2,\delta(\ell')}$ be two NPT}'', being implicit that the parameter
degree of the emerging theory is arbitrary, while the parameter degree
of the ambient theory is constrained.
\begin{rem}
The constraining on the parameter degree of the ambient theory $S_{2,\delta(\ell)}$
remarked above is basically due to the fact that in order to prove
that $S_{1,\varepsilon}$ emerges from $S_{2,\delta(\ell)}$ we will
need a functional calculus for sets of functions \emph{$C_{\ell}(P;\mathbb{K})\subset\operatorname{Map}(\operatorname{Par}(P)^{\ell};\mathbb{K})$},
where the latter $\ell$ is the parameter degree of $S_{2,\delta(\ell)}$.
This is a constrain on the nice background in which both NPT $S_{1,\varepsilon}$
and $S_{2,\delta(\ell)}$ are defined. More precisely, given $\ell\geq0$,
we say that a nice background $\mathcal{NB}$ is \emph{actioned by
degree $\ell$ parameters} if it is endowed with an action $\cdot^{\ell}:\operatorname{Par}(P)^{\ell}\times\operatorname{Nice}(E)\rightarrow\operatorname{Nice}(E)$
such that (\ref{compatibility_action}) is satisfied and which becomes
injective in the first variable when we fix the identity the second
one, i.e, $\varepsilon(\ell)\cdot^{\ell}I=\varepsilon'(\ell)\cdot^{\ell}I$
implies $\varepsilon(\ell)=\varepsilon'(\ell)$. We write $\mathcal{NB}_{\ell}$
to denote a nice background which is actioned by degree $\ell$ parameters.
We say that \emph{$C_{\ell}(P;\mathbb{K})$} is the \emph{domain of
a functional calculus} in $\mathcal{NB}_{\ell}$ if it becomes endowed
with a map $\Psi_{-}^{\ell}:C_{\ell}(P;\mathbb{K})\rightarrow\operatorname{Nice}(E)$,
assigning to each function $f$ a corresponding operator $\Psi_{f}^{\ell}$,
such that (\ref{functional_calculus}) is satisfied.
\begin{rem}
\label{injective_right_multiplication} Consider the function $r_{I}^{\ell}:\operatorname{Par}(P)^{\ell}\rightarrow\operatorname{Nice}(E)$
given by (\ref{right_multiplication}), i.e, $r_{I}^{\ell}(\varepsilon(\ell))=\varepsilon(\ell)\cdot^{\ell}I$.
The injectivity condition above means precisely that $r_{I}^{\ell}$
is injective, so that it can actually be regarded as an isomorphism
$\operatorname{Par}(P)^{\ell}\simeq\operatorname{Par}(P)^{\ell}\cdot^{\ell}I$
between its domain and its image.
\begin{rem}
\label{induced_functional_calculus}If a nice background $\mathcal{NB}$
is actioned by degre $\ell$ parameters, then it is also actioned
by degree $\ell'$ parameters, for every $0\leq\ell'\leq\ell$. In
other words, if it is of type $\mathcal{NB}_{\ell}$, then it also
is of type $\mathcal{NB}_{\ell'}$, with $0\leq\ell'\leq\ell$. This
is done by induction and noticing $\operatorname{Par}(P)^{\ell'}$
can be embedded in $\operatorname{Par}(P)^{\ell}$ by making constants
the first $\ell-\ell'$ variables. Furthermore, using the same arguments,
every functional calculus in $\mathcal{NB}_{\ell}$ with domain $C_{\ell}(P;\mathbb{K})$
induces a functional calculus in $\mathcal{NB}_{\ell'}$ with domain
$C_{\ell}^{\ell'}(P;\mathbb{K})$ given by the restriction to $\operatorname{Par}(P)^{\ell'}$
of the functions $f:\operatorname{Par}(P)^{\ell}\rightarrow\mathbb{K}$
in $C_{\ell}(P;\mathbb{K})$.
\end{rem}
\end{rem}
\end{rem}
\end{rem}
The following lemma will be the starting point for each step in the
proof of our main theorem. It says that in order to prove emergence
between two NPT it is enough to analyze emergence between the corresponding
extended operators.
\begin{lem}
\label{key_remark} Let $S_{1,\varepsilon}$ and $S_{2,\delta}$ be
two NPT of arbitrary parameter degrees and defined by operators $\Psi_{1,\varepsilon}$
and $\Psi_{2,\delta}$, respectively. If there exists a function $F$
such that $\tilde{\Psi}_{1,\varepsilon}=\tilde{\Psi}_{2,F(\varepsilon)}$
for every $\varepsilon$, then $S_{1,\varepsilon}$ emerges from $S_{2,\delta}$
\end{lem}
\begin{proof}
Indeed, since $\tilde{\cdot}$ is injective, this condition implies
$\Psi_{1,\varepsilon}=\Psi_{2,F(\varepsilon)}$. But the theories
are defined on the same nice background, so that $\ll\varphi,\Psi_{1,\varepsilon}\varphi\gg=\ll\varphi,\Psi_{2,F(\varepsilon)}\varphi\gg$
for every $\varphi\in\Gamma(E)$, which means precisely that $S_{1}[\varphi;\varepsilon]=S_{2}[\varphi;F(\varepsilon)]$
for every $\varphi$, i.e, that $S_{1,\varepsilon}=S_{2,F(\varepsilon)}$.
\end{proof}
\begin{rem}
Due to the last lemma, if $\Psi_{2,\delta}$ is a parameterized operator,
in some cases we will say ``\emph{$S_{1,\varepsilon}$ emerges from
$\Psi_{2,\delta}$}'', meaning that the parameterized operator $\Psi_{1,\varepsilon}$
emerges from $\Psi_{2,\delta}$ and, consequently, that $S_{1,\varepsilon}$
emerges from the NPT defined by $\Psi_{2,\delta}$.
\end{rem}
We close this section by showing that under coerciveness or self-adjoitness
hypothesis on the parameterized operators, the reciprocal of Lemma
\ref{key_remark} is true. This follows from the following general
fact. Let $\mathcal{H}$ be a $\mathbb{K}$-Hilbert space, $T:\mathcal{H}\rightarrow\mathcal{H}$
a bounded linear operator and consider the bilinear map $B_{T}:\mathcal{H}\times\mathcal{H}\rightarrow\mathbb{K}$
given by $B_{T}(v,w)=\langle v,T(w)\rangle$, which is bounded, since
$T$ is. For every $T$ we have a corresponding quadratic form $q_{T}:\mathcal{H}\rightarrow\mathbb{K}$
such that $q_{T}(v)=\langle v,T(v)\rangle$. Recall that $T$ is \emph{coercive}
if the induced quadratic form $q_{T}(v)=B_{T}(v,v)$ is coercive in
the classical sense, i.e, if there exists $K>0$ such that $K\Vert v\Vert^{2}\leq\vert q_{T}(v)\vert$
for every $v\in\mathcal{H}$.
\begin{lem}
\label{lax_milgran}Let $T:\mathcal{H}\rightarrow\mathcal{H}$ be
a bounded operator which is self-adjoint or coercive. Then $q_{T}\equiv0$
for every $v\in\mathcal{H}$ iff $T\equiv0$.
\end{lem}
\begin{proof}
Assume $T$ coercive, so that exists $K>0$ such that $K\Vert v\Vert^{2}\leq\vert B_{T}(v,v)\vert$.
In this case, we have $B_{T}(v,v)=0$ for every $v\in\mathcal{H}$
iff $T=0$. The ``if'' part is obvious. For the ``only if'' part,
assume $q_{T}(v)=B_{T}(v,v)=0$ for every $v$ and that $T\neq0$.
From Lax-Milgran theorem, for each $f\in\mathcal{H}^{*}$ there exists
a unique $u\in\mathcal{H}$ such that $B_{T}(u,T(v))=f(v)$ for every
$v$. In particular, taking $f=0$ we see that there exists a unique
$u$ such that $B_{T}(u,v)=0$ for every $u$. Since $T\neq0$, this
is a contradiction to the hypothesis $q_{T}\equiv0$. For the self-ajoint
case, from the spectral theorem it follows that $T$ is unitarily
equivalent to a multiplication operator $T_{\lambda}$. But $q_{T_{\lambda}}(v)=\lambda\Vert v\Vert^{2}$,
so that if $q_{T_{\lambda}}\equiv0$, then $\lambda=0$, i.e, $T\equiv0$.
\end{proof}

\section{\label{steps_1_2}First Step}

$\quad\;\,$We are now ready to prove the first step, which is also
starting point lemma, in the sense that it will be used as the basis
of many inductions.
\begin{lem}
\label{lemma_2} Let $S_{1,\varepsilon}$ be a NPT in $\mathcal{NB}_{\ell'}$,
defined by the parameterized operators $\Psi_{1,\varepsilon}$. Let
$C_{\ell'}(P;\mathbb{K})$ be a domain of a functional calculus in
$\mathcal{NB}_{\ell'}$. Then $S_{1,\varepsilon}$ emerges from every
$\Psi_{2,\delta(\ell')}^{l}=g(\delta(\ell'))\Psi^{l}$, with $l\geq0$,
where $\Psi\in\operatorname{Nice}(E)$ is right-invertible, $g\in C_{\ell'}(P;\mathbb{K})$
and $\Psi^{l}=\Psi\circ\cdots\circ\Psi$, with $\Psi^{0}=I$.
\end{lem}
\begin{proof}
By Remark \ref{key_remark} it is suffices to analyze emergence of
the extended operators. Since $I$ is right-invertible, notice that
the case $l=0$ is a particular setup of case $l=1$. Furthermore,
if $l>1$ and $\Psi$ is right-invertible, then $\Xi=\Psi^{l}$ is
right-invertible too, so that the case $l>1$ also follows from the
$l=1$ scase. Thus, we will work with $l=1$. Suppose $\tilde{\Psi}_{2,\delta(\ell')}=\tilde{\Psi}_{1,\varepsilon}$,
i.e $\widetilde{g(\delta(\ell'))I}\circ\tilde{\Psi}=\tilde{\Psi}_{1,\varepsilon}$.
Applying the right-inverse of $\tilde{\Psi}$ in both sides and using
that $\tilde{\cdot}$ is an algebra homomorphism, we find that $\widetilde{g(\delta(\ell'))I}=\tilde{\Psi}_{1,\varepsilon}\circ\widetilde{R_{\Psi}}=\widetilde{\Psi_{1,\varepsilon}\circ R_{\Psi}}$.
Since $\tilde{\cdot}$ is also injective, we get $g(\delta(\ell'))I=\Psi_{1,\varepsilon}\circ R_{\Psi}$.
Because $g\in C(\operatorname{Par}(P)^{\ell'};\mathbb{K})$ and since
$S_{2,\delta(\ell')}$ has functional calculus, from (\ref{functional_calculus})
it follows that there exists $\Psi_{g}^{\ell'}$ such that $\Psi_{g}^{\ell'}\circ(g(\delta(\ell')I)=\delta(\ell')\cdot^{\ell'}I$,
so that $\Psi_{g}^{\ell'}\circ(\Psi_{1,\varepsilon}\circ R_{\Psi})=\delta(\ell')\cdot^{\ell'}I$.
Thus, we have a function $F:\operatorname{Par}(P)^{\ell}\rightarrow\operatorname{Nice}(E)$,
given by $F(\varepsilon)=\Psi_{g}^{\ell'}\circ(\Psi_{1,\varepsilon}\circ R_{\Psi})$,
which by the above is exactly $r_{I}^{\ell'}(\delta(\ell'))$. Thus,
due Remark \ref{injective_right_multiplication} $F$ can be regarded
as a map $F:\operatorname{Par}(P)^{\ell}\rightarrow\operatorname{Par}(P)^{\ell'}$
which by construction is such that $\tilde{\Psi}_{1,\varepsilon}=\tilde{\Psi}_{2,F(\varepsilon)}$,
as desired.
\end{proof}

\section{\label{steps_3_4}Second and Third Steps}

$\quad\;\,$The next two steps are additivity and multiplicativity
results for emergence phenomena. In order to prove them, we have to
add hypothesis on the set of parameters instead of on the shape of
the operators. A \emph{nowhere vanishing space} is a subset $W\subset V$
of a $\mathbb{K}$-vector space $V$ such that $v+v'\in W$ for every
$v,v'\in W$ with $v'\neq-v$. In particular, $0\notin W$. We also
require that $c\cdot v\in W$ for every scalar $c\neq0$ and for every
nonzero vector $v\in W$. Let $W$ be a nowhere vanishing space and
let $Z$ be a vector space or nowhere vanishing space. A function
$T:W\rightarrow Z$ is \emph{linear }if it preserves sum and scalar
multiplication. Notice that $T$ is a nowhere vanishing function.
Bilinear maps are defined analogously. A \emph{nowhere vanishing algebra}
is a nowhere vanishing space $A\subset V$ endowed with a bilinear
multiplication $*:W\times W\rightarrow W$. Let $A$ be nowhere vanishing
algebra and let $B$ be another nowhere vanishing algebra or an algebra
in the classical sense. An \emph{homomorphism }between them is a linear
map $T:A\rightarrow B$ preserving the multiplication.

Given $\ell\geq0$, we say that a nice backgroud $\mathcal{NB}$ has
a \emph{degree $\ell$ nowhere vanishing space of parameters }(resp.
\emph{degree $\ell$ vector space of parameters}) if $\operatorname{Par}(P)^{\ell}$
is a nowhere vanishing space (resp. vector space). Similarly, we say
that $\mathcal{NB}$ has a \emph{degree $\ell$ nowhere vanishing
algebra of parameters }(resp. \emph{degree $\ell$ algebra of parameters})
if $\operatorname{Par}(P)^{\ell}$ is a nowhere vanishing algebra
(resp. algebra). We say that a NPT $S_{\varepsilon(\ell)}$ of parameter
degree $\ell$ and parameterized operators $\Psi_{\varepsilon(\ell)}$
is \emph{partially additive on parameters }(resp. \emph{additive on
parameters}) if the underlying nice background $\mathcal{NB}$ has
degree $\ell$ nowhere vanishing space of parameters\emph{ }(resp.
degree $\ell$ vector space of parameters) and if the rule $\varepsilon(\ell)\mapsto\Psi_{\varepsilon(\ell)}$
is linear relative to that structure, i.e, $\Psi_{\varepsilon(\ell)+\varepsilon(\ell)'}=\Psi_{\varepsilon(\ell)}+\Psi_{\varepsilon(\ell)'}$
and $\Psi_{c\varepsilon(\ell)}=c\Psi_{\varepsilon(\ell)}$. In an
analogous way, we say that $S_{\varepsilon(\ell)}$ is \emph{partially
multiplicative on parameters }(resp. \emph{multiplicative on parameters})
if $\mathcal{NB}$ has a degree $\ell$ nowhere vanishing algebra
of parameters (resp. a degree $\ell$ algebra of parameters) and the
rule $\varepsilon(\ell)\mapsto\Psi_{\varepsilon(\ell)}$ is not necessarily
linear, but preserves the multiplication, i.e., $\Psi_{\varepsilon(\ell)*\varepsilon'(\ell)}=\Psi_{\varepsilon(\ell)}\circ\Psi_{\varepsilon'(\ell)}$.
Finally, we say that $S_{\varepsilon(\ell)}$ is \emph{partially homomorphic}
\emph{on parameters }(resp. \emph{homomorphic on parameters}) it is
both partially additive and partially multiplicative (resp. additive
and multiplicative).
\begin{example}
It is straighforward to check that if $\mathcal{NB}$ has a degree
$\ell$ nowhere vanishing space of parameters (resp. degree $\ell$
vector space of parameters), then it also has for every $\ell'=n\ell$,
where $n\geq1$ is another integer. The structure is obtained by noticing
that 
\[
\operatorname{Par}(P)^{n\ell}=\operatorname{Par}(P)^{\ell}\times...\times\operatorname{Par}(P)^{\ell}
\]
and then defining sum and multiplication componentwise. A similar
argument applies to the case of degree $\ell$ nowhere vanishing algebra
of parameters (resp. degree $\ell$ algebra of parameters). 
\begin{example}
In general $P$ is a vector bundle or an algebra bundle, so that $\Gamma(P)$
is a vector space or an algebra and $\operatorname{Par}(P)\subset\Gamma(P)$
is some nowhere vanishing subspace, vector subspace, nowhere vanishing
algebra or subalgebra. By the last example it then follows that each
of these structures can be lifted to $\operatorname{Par}(P)^{n}$
for any $n\geq1$.
\end{example}
\end{example}
Given two NPT theories $S_{2,\delta(\ell')}$ and $S_{3,\kappa(\ell'')}$,
define their \emph{sum} as the theory $S_{2,\delta(\ell')}+S_{3,\kappa(\ell'')}$
whose parameter and whose parameterized operator is $\Psi_{\delta(\ell'),\kappa(\ell'')}=\Psi_{2,\delta(\ell')}+\Psi_{3,\kappa(\ell'')}$.
In an analogous way we define the \emph{composition} \emph{theory}
$S_{2,\delta(\ell')}\circ S_{3,\kappa(\ell'')}$. We say that $S_{1,\varepsilon(\ell)}$
\emph{emerges} from $S_{2,\delta(\ell')}+S_{3,\kappa(\ell'')}$ (resp.
$S_{2,\delta(\ell')}\circ S_{3,\kappa(\ell'')}$) if there exists
$F_{+}:\operatorname{Par}(P)^{\ell}\rightarrow\operatorname{Par}(P)^{\ell'}\times\operatorname{Par}(P)^{\ell''}$
(resp. $F_{\circ}$) such that 
\[
S_{1}[\varphi;\varepsilon(\ell)]=(S_{2}+S_{3})[\varphi;F_{+}(\varepsilon(\ell))]\quad\text{(resp.}\;S_{1}[\varphi;\varepsilon(\ell)]=(S_{2}\circ S_{3})[\varphi;F_{\circ}(\varepsilon(\ell))]\text{)},
\]
exactly as in the previous situations. The next lemmas are independent
of the parameter degrees of the theories $S_{2,\delta(\ell')}$ and
$S_{3,\delta(\ell'')}$. Thus, following Remark \ref{remark_notations},
they will be omitted. 
\begin{lem}
\label{lemma_3} Let $S_{1,\varepsilon(\ell)}$, $S_{2,\delta}$ and
$S_{3,\kappa}$ be three NPT, defined by $\Psi_{1,\varepsilon(\ell)}$,
$\Psi_{2,\delta}$ and $\Psi_{3,\kappa}$. Assume that:
\begin{enumerate}
\item the theory $S_{1,\varepsilon(\ell)}$ is additive or partially additive
on parameters;
\item the theory $S_{1,\varepsilon(\ell)}$ emerges from both $S_{2,\delta}$
and $S_{3,\kappa}$, while theory $S_{2,\delta}$ emerges from $S_{3,\kappa}$;
\end{enumerate}
Then $S_{1,\varepsilon(\ell)}$ emerges from the sum $S_{2,\delta}+S_{3,\kappa}$.
If, in addition, $S_{3,\kappa}$ is also additive or partially additive
on parameters, then the emergence of $S_{1,\varepsilon(\ell)}$ from
$S_{2,\delta}+S_{3,\kappa}$ is equivalent to a specific new emergence
of $S_{1,\varepsilon(\ell)}$ from $S_{3,\kappa}$ .
\end{lem}
\begin{proof}
Assume first that $S_{1,\varepsilon(\ell)}$ is additive. From the
second hypothesis we conclude that $\Psi_{1,\varepsilon(\ell)}=\Psi_{2,F(\varepsilon(\ell))}$,
$\Psi_{1,\varepsilon(\ell)}=\Psi_{3,G(\varepsilon)}$ and $\Psi_{2,\delta}=\Psi_{3,H(\delta)}$
for certain functions $F,G,H$. Summing the first two of these conditions
and using the third one with $\delta=F(\varepsilon(\ell))$ we find
\begin{equation}
2\Psi_{1,\varepsilon(\ell)}=\Psi_{3,H(F(\varepsilon(\ell)))}+\Psi_{3,G(\varepsilon(\ell))}.\label{lemma_additive_1}
\end{equation}
On the other hand, since $S_{1,\varepsilon(\ell)}$is additive, we
get
\begin{equation}
\Psi_{1,\varepsilon(\ell)}=\Psi_{3,H(F(\varepsilon(\ell)/2))}+\Psi_{3,G(\varepsilon(\ell)/2)},\label{lemma_additive_2}
\end{equation}
so that $S_{1}[\varphi;\varepsilon(\ell)]=(S_{2}+S_{3})[\varphi;K(\varepsilon(\ell))]$
with $K(\varepsilon(\ell))=(G(\varepsilon(\ell)/2),H(F(\varepsilon(\ell)/2)))$,
finishing the first part of the proof in the additive case. For the
second part, assume that $S_{3,\kappa}$ is additive on parameters.
In this case, the right-hand side of the expression above is equivalent
to $\Psi_{3,H(F(\varepsilon(\ell)/2))+G(\varepsilon(\ell)/2)}$, meaning
that $S_{1}[\varphi;\varepsilon(\ell)]=S_{3}[\varphi;L(\varepsilon(\ell))]$
with $L(\varepsilon(\ell))=H(F(\varepsilon(\ell)/2))+G(\varepsilon(\ell)/2)$.
Now, observe that nowhere in the proof have we used that the parameter
space contains the null vector or opposite vectors. This means that
the same arguments work equally well in the partially additive setting.
\end{proof}
\begin{rem}
\label{remark_emergence_scalar}One can similarly show that if $S_{1,\varepsilon(\ell)}$
is additive or partially additive which emerges from $S_{2,\delta}$,
then it emerges from $c\cdot S_{2,\delta}$ for every $c\neq0$. Indeed,
if $\Psi_{1,\varepsilon(\ell)}=\Psi_{2,F(\varepsilon(\ell))}$ and
if $\Psi_{2,F(\varepsilon(\ell))}=\frac{c}{c}\Psi_{2,F(\varepsilon(\ell))}=c\Psi_{2,F(\varepsilon(\ell))/c}$,
then $\Psi_{1,\varepsilon}=c\Psi_{2,F(\varepsilon(\ell))/c}$. Reciprocally,
if $S_{1,\varepsilon(\ell)}$ emerges from $c\cdot S_{2,\delta}$,
with $c\neq0$, then it also emerges from $S_{2,\delta}$.
\end{rem}
Given $\ell\geq0$, we say that a nice background $\mathcal{NB}$
has \emph{degree $\ell$} \emph{square roots }if it has a degree $\ell$
nowhere vanishing algebra of parameters or degree $\ell$ algebra
of parameters and if for every $\varepsilon(\ell)$ there exists some
$\varepsilon(\ell)^{1/2}$ such that $(\varepsilon(\ell)^{1/2})^{2}\equiv\varepsilon(\ell)^{1/2}*\varepsilon(\ell)^{1/2}=\varepsilon(\ell)$.
We write $\mathcal{NB}^{\ell}$ to denote this fact.
\begin{example}
\label{example_higher_degree_square_roots} From Remark \ref{induced_functional_calculus},
if $\mathcal{NB}$ has a degree $\ell$ nowhere vanishing algebra
of parameters or degree $\ell$ algebra of parameters, then it also
has for every $\ell'=n\ell$, with $n\geq1$. Relatively to this componentwise
structure, one can show that if $\mathcal{NB}$ has degree $\ell$
square roots, then it also has for every $\ell'=n\ell$, i.e, if $\mathcal{NB}$
is of type $\mathcal{NB}^{\ell}$, then it also is of type $\mathcal{NB}^{n\ell}.$
\end{example}
We can now work on the third step.
\begin{lem}
\label{lemma_4}Let $S_{1,\varepsilon(\ell)}$, $S_{2,\delta}$ and
$S_{3,\kappa}$ be three NPT with parameterized operators $\Psi_{1,\varepsilon(\ell)}$,
$\Psi_{2,\delta}$ and $\Psi_{3,\kappa}$, as above. Assume that:
\begin{enumerate}
\item the theory $S_{1,\varepsilon(\ell)}$ is multiplicative or partially
multiplicative on parameters;
\item the theory $S_{1,\varepsilon(\ell)}$ emerges from both $S_{2,\delta}$
and $S_{3,\kappa}$, while theory $S_{2,\delta}$ emerges from $S_{3,\kappa}$;
\item the underlying nice background has degree $\ell$ square roots.
\end{enumerate}
Then $S_{1,\varepsilon(\ell)}$ emerges from the composition $S_{2,\delta}\circ S_{3,\kappa}$.
If in addition, $S_{3,\kappa}$ is also multiplicative or partially
multiplicative on parameters, then the emergence of $S_{1,\varepsilon(\ell)}$
from $S_{2,\delta}\circ S_{3,\kappa}$ is equivalent to a specific
new emergence of $S_{1,\varepsilon(\ell)}$ from $S_{3,\kappa}$.
\end{lem}
\begin{proof}
We will work only in the multiplicative case. The partially multiplicative
one will follows from the same argument used in Lemma \ref{lemma_3}.
The first part of the proof also follows the same lines of Lemma \ref{lemma_3},
the only difference being that equation (\ref{lemma_additive_1})
is now replaced with 
\[
\Psi_{1,\varepsilon(\ell)}\circ\Psi_{1,\varepsilon(\ell)}=\Psi_{3,H(F(\varepsilon(\ell)))}\circ\Psi_{3,G(\varepsilon(\ell))}=\Psi_{3,G(\varepsilon(\ell))}\circ\Psi_{3,H(F(\varepsilon(\ell)))}.
\]
Since $S_{1,\varepsilon(\ell)}$ is multiplicative and since the background
$\mathcal{NB}^{\ell}$ has degree $\ell$ square roots we can write
an analogue for (\ref{lemma_additive_2}):
\[
\Psi_{1,\varepsilon(\ell)}=\Psi_{3,H(F(\varepsilon(\ell)^{1/2}))}\circ\Psi_{3,G(\varepsilon(\ell)^{1/2})}=\Psi_{3,G(\varepsilon(\ell)^{1/2})}\circ\Psi_{3,H(F(\varepsilon(\ell)^{1/2}))},
\]
so that $S_{1}[\varphi,\varepsilon(\ell)]=(S_{2}\circ S_{3})[\varphi;K(\varepsilon(\ell))]$
with $K(\varepsilon(\ell))=(G(\varepsilon(\ell)^{1/2}),H(F(\varepsilon(\ell)^{1/2})))$,
finishing this first part. For the second part, assuming $S_{3,\kappa}$
multiplicative on parameters, just notice that the right-hand side
of the above expression becomes 
\[
\Psi_{3,G(\varepsilon(\ell)^{1/2})*H(F(\varepsilon(\ell)^{1/2}))}=\Psi_{3,H(F(\varepsilon(\ell)^{1/2}))*G(\varepsilon(\ell)^{1/2})},
\]
finishing the proof. 
\end{proof}
\begin{cor}
\label{corollary_powers}Let $S_{\varepsilon(\ell)}$ be a multiplicative
or partially multiplicative NPT defined in $\mathcal{NB}^{\ell}$
and with operator $\Psi_{\varepsilon(\ell)}$. Then, for every $l,m\geq1$,
the theories $S_{\varepsilon(\ell)}^{l}$ and $S_{\varepsilon(\ell)}^{m}$
emerges each one from the other, where $S_{\varepsilon(\ell)}^{i}=S_{\varepsilon(\ell)}\circ\cdots\circ S_{\varepsilon(\ell)}$.
\end{cor}
\begin{proof}
Fixed $m=1$, use induction on $l$. The base of induction is the
fact that a theory always emerges from itself. For the induction step,
use previous lemma. This implies that $S_{\varepsilon(\ell)}$ emerges
from $S_{\varepsilon(\ell)}^{l}$ for every $l\geq1$. Consequently,
$S_{\varepsilon(\ell)}^{m-1}\circ S_{\varepsilon(\ell)}$ emerges
from $S_{\varepsilon(\ell)}^{m-1}\circ S_{\varepsilon(\ell)}^{l}$
for every $m,l\geq1$, where, by definition $S_{\varepsilon(\ell)}^{0}$
is the theory whose operator is the identity. Since $S_{\varepsilon(\ell)}^{i}\circ S_{\varepsilon(\ell)}^{j}=S_{\varepsilon(\ell)}^{i+j}$,
we conclude that $S_{\varepsilon(\ell)}^{m}$ emerges from $S_{\varepsilon(\ell)}^{m-1+l}=S_{\varepsilon(\ell)}^{l(m)}$,
with $m\leq l(m)$. On the other hand, since the identity is an invertible
map, we see that $S_{\varepsilon(\ell)}^{l(m)}$ also emerges from
$S_{\varepsilon(\ell)}^{m}$. 
\end{proof}
Recall that if a nice background $\mathcal{NB}$ is actioned by degree
$\ell'$ parameters, we write $\mathcal{NB}_{\ell'}$ to denote this
fact. Furthermore, by the above, if $\mathcal{NB}$ has degree $\ell$
square roots, we write $\mathcal{NB}^{\ell}$. Thus, from now on,
if $\mathcal{NB}$ has both properties we will write $\mathcal{NB}_{\ell'}^{\ell}$.

\section{\label{step_4_1}Fourth Step}

$\quad\;\,$As a consequence of the previous lemmas we can prove the
fourth step. We say that a functional calculus in a nice background
$\mathcal{NB}_{\ell''}$ is \emph{unital} if its domain $C_{\ell''}(P;\mathbb{K})$
contains the constant function $f\equiv1$.
\begin{lem}
\label{lemma_4_1} Let $S_{\varepsilon(\ell)}$ be a homomorphic or
partially homomorphic on parameters NPT, defined in a nice background
$\mathcal{NB}_{\ell''}^{\ell}$ and whose parameterized operator is
$\Psi_{\varepsilon(\ell)}$. Let $C_{\ell''}(P;\mathbb{K})$ be the
domain of a unital functional calculus in $\mathcal{NB}_{\ell''}^{\ell}$.
Then $S_{\varepsilon(\ell)}$ emerges from every theory whose operator
is of the form 
\[
p_{\ell'}^{l}[\Psi]=\sum_{i=1}^{l}m_{\delta_{i}(\ell')}^{i}[\Psi]=\sum_{i=1}^{l}f_{i}(\delta_{i}(\ell'))\Psi^{i},
\]
where $\ell''=l\ell'$, $f_{i}\in C_{\ell''}^{\ell'}(P;\mathbb{K})$
for $i=1,...,l$ and such that $\Psi$ is right-invertible\footnote{The condition on $f_{i}$ makes sense due Remark \ref{induced_functional_calculus}. }.
\end{lem}
\begin{proof}
For each $j=1,...,l$, let $\Gamma_{j}=\sum_{i=j}^{l}f_{i}(\delta_{i}(\ell'))\Psi^{i-1}$
and notice that 
\[
\sum_{i=1}^{l}\Psi_{i,\delta_{i}(\ell')}=(\sum_{i=1}^{l}f_{i}(\delta_{i}(\ell'))\Psi^{i-1})\circ\Psi=\Gamma_{1}(\delta_{I}(\ell'))\circ\Psi.
\]
Since $\Psi$ is right-invertible and since the nice background $\mathcal{NB}_{\ell''}^{\ell}$
has unital functional calculus, from Lemma \ref{lemma_2} it follows
that $S_{\varepsilon(\ell)}$ emerges from $1\cdot\Psi$. On the other
hand, again from Lemma \ref{lemma_2} we see that the theory defined
by the operator $\Gamma_{1}$ also emerges from that defined by $1\cdot\Psi$.
Thus, if $S_{\varepsilon(\ell)}$ itself emerges from $\Gamma_{1}$
one can use Lemma \ref{lemma_4} to conclude that it actually emerges
from $\Gamma_{1}\circ\Psi$. In turn, notice that $\Gamma_{1}=f_{1}\cdot I+\Gamma_{2}\circ\Psi=\Gamma_{2}\circ\Psi+f_{1}\cdot I$.
But, since $I$ is right-invertible and since $f_{1}\in C_{\ell''}^{\ell'}(P;\mathbb{K})$,
from Remark \ref{induced_functional_calculus} and from Lemma \ref{lemma_2}
we get that $S_{\varepsilon(\ell)}$ emerges from the theory defined
by $f_{1}\cdot I$, while by the same argument we see that $\Gamma_{2}\circ\Psi$
emerges from $f_{1}\cdot I$. Therefore, if $S_{\varepsilon(\ell)}$
emerges from $\Gamma_{2}\circ\Psi$ we will be able to use Lemma \ref{lemma_3}
to conclude that it emerges from $\Gamma_{1}$, finishing the proof.
It happens that, as done for $\Gamma_{1}\circ\Psi$, we see that $\Gamma_{2}$
emerges from $\Psi$ and we already know that $S_{\varepsilon(\ell)}$
emerges from $\Psi$. Thus, our problem is to prove that $S_{\varepsilon(\ell)}$
emerges from $\Gamma_{2}$ instead of from $\Gamma_{1}$. A finite
induction argument proves that if $S_{\varepsilon(\ell)}$ emerges
from $\Gamma_{l}$, then it emerges from $\Gamma_{j}$, for each $j=1,...,l$.
Recall that $\Gamma_{l}=f_{l}\cdot\Psi^{l-1}$. Since $f_{l}\in C_{\ell''}^{\ell'}(P;\mathbb{K})$
we can use Lemma \ref{lemma_2} to see that $S_{\varepsilon(\ell)}$
really emerges from $\Gamma_{l}$.
\end{proof}
Let $\operatorname{Map}(\operatorname{Par}(P)^{\ell'};\mathbb{K})$
be the algebra of functions $f:\operatorname{Par}(P)^{\ell'}\rightarrow\mathbb{K}$
and let $\operatorname{Map}(\operatorname{Par}(P)^{\ell'};\mathbb{K})[x]$
be the corresponding polynomial algebra. Recall that any polynomial
ring has a grading by the degree, so that we can write 
\[
\operatorname{Map}(\operatorname{Par}(P)^{\ell'};\mathbb{K})[x]\simeq\bigoplus_{l\geq0}\operatorname{Map}_{l}(\operatorname{Par}(P)^{\ell'};\mathbb{K})[x],
\]
where the right-hand side consists of the sum over the $\mathbb{K}$-vector
spaces of polynomials with fixed degree $l$. By extension of scalars
$\operatorname{Nice}(E)$ can be regarded as an algebra over $\operatorname{Map}(\operatorname{Par}(P)^{\ell'};\mathbb{K})$,
denoted by $\operatorname{Nice}_{\mathbb{K}}^{\ell'}(E)$, so that
for every $l\geq1$ we have an evaluation map
\begin{equation}
ev_{\ell';l}:\operatorname{Map}_{l}(\operatorname{Par}(P)^{\ell'};\mathbb{K})[x]\times\operatorname{Nice}(E)\rightarrow\operatorname{Nice}_{\mathbb{K}}^{\ell'}(E)\label{evaluation}
\end{equation}
which takes a polynomial $p_{\ell'}^{l}[x]=\sum_{i=0}^{l}f_{i}\cdot x^{i}$
and an operator $\Psi$ and produces $ev_{l}(p_{\ell'}^{l}[x],\Psi)=p_{\ell'}^{l}[\Psi]=\sum_{i=0}^{l}f_{i}\cdot\Psi^{i}$,
where $\Psi^{i}=\Psi\circ...\circ\Psi$. Let $\operatorname{Poly}_{\mathbb{K};l}^{\ell'}(E)$
denote the image of $ev_{\ell';l}$. If one fixes an operator $\Psi$
such that $\Psi^{0},\Psi^{1},\Psi^{2},\cdots$ are linearly independent
as objects of $\operatorname{Nice}(E)$ the map $ev_{\ell';l}$ becomes
injective. More precisely, we have the following result:
\begin{prop}
Let $\Psi$ be an operator such that $\Psi^{i}$ is linearly independent
of $\Psi^{j}$ for every $1\leq i,j\leq l$, for some $l\geq1$. Then
the evaluation morphism at $\Psi$
\begin{equation}
ev_{\ell';\leq l}^{\Psi}:\operatorname{Map}_{\leq l}(\operatorname{Par}(P)^{\ell'};\mathbb{K})[x]\rightarrow\operatorname{Nice}_{\mathbb{K}}^{\ell'}(E)\label{evaluation_proposition}
\end{equation}
is injective, where the left-hand side is the subspace of polynomials
of degree $d\leq l$.
\end{prop}
\begin{proof}
The kernel of $ev_{\ell';\leq l}^{\Psi}$ consists of those $p_{\ell'}^{d}[x]=\sum_{i\leq d}f_{i}\cdot x^{i}$
such that $ev_{\ell';\leq l}^{\Psi}(p_{\ell'}^{d}[x])=0=\sum_{i\leq d}f_{i}\cdot\Psi^{i}$.
Due to the linearly independence hypothesis, this is the case iff
$f_{i}=0$, i.e, $p_{\ell'}^{d}[x]=0$.
\end{proof}
Observe that we have another evaluation 
\[
ev_{\ell';l}^{l}:\operatorname{Poly}_{\mathbb{K};l}^{\ell'}(E)\times[\operatorname{Par}(P)^{\ell'}]^{l}\rightarrow\operatorname{Nice}(E)
\]
such that $ev_{\ell';l}^{l}(p_{\ell'}^{l}[\Psi],\delta_{I}(\ell'))=p_{\ell'}^{l}[\Psi](\delta_{I}(\ell'))=\sum_{i=0}^{l}f_{i}(\delta_{i}(\ell'))\Psi^{i}$,
where $\delta_{I}(\ell')=(\delta_{1}(\ell'),...,\delta_{l}(\ell'))$.
Therefore, for every $l\geq1$ we have a complete evaluation 
\[
ev_{\ell'}^{l}:\operatorname{Map}_{l}(\operatorname{Par}(P)^{\ell'};\mathbb{K})[x]\times\operatorname{Nice}(E)\times[\operatorname{Par}(P)^{\ell'}]^{l}\rightarrow\operatorname{Nice}(E)
\]
given by $ev_{\ell'}^{l}(p_{\ell'}^{l}[x],\Psi,\delta_{I}(\ell'))=ev_{\ell';l}^{l}(ev_{\ell';l}(p_{\ell'}^{l}[x],\Psi),\delta_{I}(\ell'))$.
Let $\underline{\operatorname{Nice}(E)}$ be the subset of operators
$\Psi\in\operatorname{Nice}(E)$ whose extension $\tilde{\Psi}\in B(\mathcal{H}(E))$
is right-invertible, let $C_{\ell'}(P;\mathbb{K})\subset\operatorname{Map}(\operatorname{Par}(P)^{\ell'};\mathbb{K})$
be a subset and let $C_{\ell;l}(P;\mathbb{K})$ the set of polynomials
$p_{\ell'}^{l}[x]=\sum f_{i}\cdot x^{i}$ of degree $l$ whose coefficients
belongs to $C_{\ell'}(P;\mathbb{K})$. Restricting (\ref{evaluation}),
we get a map 
\[
\underline{ev}_{\ell';l}^{C}:C_{\ell';l}(P;\mathbb{K})\times\underline{\operatorname{Nice}}(E)\rightarrow\operatorname{Nice}_{\mathbb{K}}^{\ell'}(E).
\]

\begin{defn}
\emph{\label{def_type_I}Let $\mathcal{NB}^{\ell}$ be a nice background
with degree $\ell$ square roots. Let us say that the set of polynomials
$C_{\ell';l}(P;\mathbb{K})$ is }coherent\emph{ in $\mathcal{NB}^{\ell}$
if}
\begin{enumerate}
\item \emph{$\mathcal{NB}^{\ell}$ is actioned by degree $\ell''=l\ell'$
parameters (i.e, it is of the form $\mathcal{NB}_{\ell''}^{\ell}$);}
\item \emph{there exists a set $C_{\ell''}(P;\mathbb{K})\subset\operatorname{Map}(\operatorname{Par}(P)^{\ell''};\mathbb{K})$
which is the domain of unital functional calculus in $\mathcal{NB}_{\ell''}^{\ell}$
and such that $C_{\ell'}(P;\mathbb{K})=C_{\ell''}^{\ell'}(P;\mathbb{K})$.}
\end{enumerate}
\end{defn}
After this discussion we see that Lemma \ref{lemma_4_1} (i.e, the
first version of the fourth step) can be rewritten as following:
\begin{lem}[Lemma \ref{lemma_4_1} revisited]
\label{lemma_4_1_revised}  Let $S_{\varepsilon(\ell)}$ be a homomorphic
or partially homomorphic on parameters NPT, defined in nice background
$\mathcal{NB}_{\ell''}^{\ell}$. Then it emerges from any theory in
the image of $\underline{ev}_{\ell';l}^{C}$, where $\ell''=l\ell'$
and $C_{\ell';l}(P;\mathbb{K})$ is coherent in $\mathcal{NB}_{\ell''}^{\ell}$.
\end{lem}
\begin{rem}
From Remark \ref{induced_functional_calculus} it follows that if
$C_{\ell';l}(P;\mathbb{K})$ is coherent in $\mathcal{NB}^{\ell}$,
then $C_{\ell';l'}(P;\mathbb{K})$ is also coherent for every $l'\leq l$. 
\end{rem}

\section{Fifth Step \label{step_5}}

$\quad\;\,$The final step before proving the main result is a recurrence
lemma which is obtained as a consequence of Lemma \ref{lemma_3}-\ref{lemma_4}.
Recall that $\operatorname{Nice}(E)$ is an algebra and, as in any
algebra $A$ we can ask if a given element $a\in A$ divides from
the left (resp. from the right) another element $b\in A$, meaning
that there exists $q\in A$, called the \emph{quotient between $b$
and $a$}, such that $b=a*q$ (resp $b=q*a$). In $\operatorname{Nice}(E)$,
given two operators $\Psi_{2}$ and $\Psi_{3}$, this means that there
exists a third $Q\in\operatorname{Nice}(E)$ such that $\Psi_{2}=\Psi_{3}*Q$
(resp. $\Psi_{2}=Q*\Psi_{3}$). Given two NPT $S_{2,\delta(\ell')}$
and $S_{3,\kappa(\ell'')}$, with respective parameterized operators
$\Psi_{2,\delta(\ell')}$ and $\Psi_{3,\kappa(\ell'')}$, let us say
that $S_{2,\delta(\ell')}$ is \emph{divisible from the left (resp.
from the right)} by $S_{3,\kappa(\ell'')}$ if for every $\delta(\ell')$
there exists $\kappa(\ell'')$ such that $\Psi_{3,\kappa(\ell'')}$
divides $\Psi_{2,\delta(\ell')}$ from the left (resp. from the right),
so that $\Psi_{2,\delta(\ell')}=\Psi_{3,\kappa(\ell'')}\circ Q_{\delta(\ell'),\kappa(\ell'')}$
(resp. $\Psi_{2,\delta(\ell')}=Q_{\delta(\ell'),\kappa(\ell'')}\circ\Psi_{3,\kappa(\ell'')})$. 
\begin{lem}
\label{lemma_5} Let $S_{1,\varepsilon(\ell)}$ be a homomorphic or
partially homomorphic NPT defined in a nice background $\mathcal{NB}_{\ell'}^{\ell}$
and whose parameterized operator is $\Psi_{1,\varepsilon(\ell)}$.
Let Given $l\geq1$, let $S_{j,\delta_{j}(\ell')}$ and $S_{k,\kappa_{k}(\ell')}$,
with $1\leq j,k\leq l$ be two families of NPT, also defined in $\mathcal{NB}_{\ell''}^{\ell}$.
Let $C_{\ell''}(P;\mathbb{K})$ be the domain of a functional calculus
in $\mathcal{NB}_{\ell''}^{\ell}$. Assume that:
\begin{enumerate}
\item $S_{1,\varepsilon(\ell)}$ emerges from $S_{2_{j},\delta_{j}(\ell')}$
and from $S_{3_{k},\kappa_{k}(\ell')}$ for every $j,k$;
\item $S_{2_{j},\delta_{j}(\ell')}$ emerges from $S_{3_{k},\kappa_{k}(\ell')}$
is $k=j$;
\item for every $2\leq k\leq l$ the theory $S_{3_{k},\kappa_{k}(\ell')}$
is divisible from the right by a monomial $m_{\kappa_{k}(\ell')}^{d(k)}=g_{k}(\kappa_{k}(\ell'))\Psi^{d(k)}$,
where $\Psi$ is right-invertible and $g_{k}\in C_{\ell''}^{\ell'}(P;\mathbb{K})$
so that $\Psi_{3_{k},\kappa_{k}(\ell')}=Q_{k}\circ m_{\kappa_{k}(\ell')}^{d(k)}$;
\item for every $1\leq m\leq l-1$ the theories $S_{\delta_{J},\kappa_{J}}^{m}=\sum_{j=1}^{m}(S_{2_{j},\delta_{j}}\circ S_{3_{j},\kappa_{j}})$
and $S_{2_{m+1},\delta_{m+1}}$ emerges from $Q_{m+1}$;
\item for every $1\leq m\leq l-1$ the theory $S_{\delta_{J},\kappa_{J}}^{m}$
emerges from $S_{2_{m+1},\delta_{m+1}}$.
\end{enumerate}
Then $S_{1,\varepsilon(\ell)}$ emerges from $S_{\delta_{J},\kappa_{J}}^{l}$.
\end{lem}
\begin{proof}
We proceed by induction in $l$. First of all, notice that from the
first two hypotheses and from Lemma \ref{lemma_4} we see that $S_{1,\varepsilon(\ell)}$
emerges from the composition $S_{2_{j},\delta_{j}}\circ S_{3_{j},\kappa_{j}}$
for every $j=1,...,l$. In particular, it emerges from $S_{\delta_{1},\kappa_{1}}^{1}=S_{2_{1},\delta_{1}}\circ S_{3_{1},\kappa_{1}}$,
which is the base of induction. For every $m=1,...,l-1$ it also emerges
from $S_{2_{m+1},\delta_{m+1}}\circ S_{3_{m+1},\kappa_{m+1}}$. For
the induction step, suppose that $S_{1,\varepsilon(\ell)}$ emerges
from $S_{\delta_{J},\kappa_{J}}^{m}=\sum_{j=1}^{m}S_{2_{j},\delta_{j}}\circ S_{3_{j},\kappa_{j}}$
for every $1\leq m\leq l-1$ and let us show that it emerges from
$S_{\delta_{J},\kappa_{J}}^{m+1}$. Notice that 
\begin{eqnarray*}
S_{\delta_{J},\kappa_{J}}^{m+1} & = & \sum_{j=1}^{m+1}S_{2_{j},\delta_{j}}\circ S_{3_{j},\kappa_{j}}=\sum_{j=1}^{m}(S_{2_{j},\delta_{j}}\circ S_{3_{j},\kappa_{j}})+S_{2_{m+1},\delta_{m+1}}\circ S_{3_{m+1},\kappa_{m+1}}\\
 & = & S_{\delta_{J},\kappa_{J}}^{m}+S_{2_{m+1},\delta_{m+1}}\circ S_{3_{m+1},\kappa_{m+1}}.
\end{eqnarray*}
From the induction hypothesis $S_{1,\varepsilon(\ell)}$ emerges from
$S_{\delta_{J},\kappa_{J}}^{m}$, while by the above it also emerges
from $S_{2_{m+1},\delta_{m+1}}\circ S_{3_{m+1},\kappa_{m+1}}$. Thus,
if $S_{\delta_{J},\kappa_{J}}^{m}$ emerges from $S_{2_{m+1},\delta_{m+1}}\circ S_{3_{m+1},\kappa_{m+1}}$
we can use Lemma \ref{lemma_3} to conclude that $S_{1,\varepsilon(\ell)}$
emerges from $S_{\delta_{J},\kappa_{J}}^{m+1}$. From Lemma \ref{lemma_2}
we know that $Q_{m+1}$, $S_{\delta_{J},\kappa_{J}}^{m}$ and $S_{2_{m+1},\delta_{m+1}}$
emerge from $m_{\kappa_{m_{+1}}(\ell')}^{d(m+1)}$. Due to the fourth
hypothesis, by Lemma \ref{lemma_4} we see that $S_{2_{m+1},\delta_{m+1}}$
and $S_{\delta_{J},\kappa_{J}}^{m}$ emerge from the composition $S_{3_{m+1},\kappa_{m+1}(\ell')}=Q_{m+1}\circ m_{\kappa_{m+1}(\ell')}^{d(m+1)}$.
Therefore, if we prove that $S_{\delta_{J},\kappa_{J}}^{m}$ also
emerges from $S_{2_{m+1},\delta_{m+1}}$, then Lemma \ref{lemma_4}
will imply that it emerges from $S_{2_{m+1},\delta_{m+1}}\circ S_{3_{m+1},\kappa_{m+1}}$,
as desired. But this remaining condition is precisely the fifth hypothesis. 
\end{proof}

\section{The Theorem \label{sec_theorem}}

$\quad\;\,$Let $r\geq1$ be a positive integer and let 
\[
\operatorname{Map}(\operatorname{Par}(P)^{\ell'};\mathbb{K})[x_{1},...,x_{r}]
\]
be the polynomial ring in variables $x_{1},...,x_{r}$. In analogy
to (\ref{evaluation}) we have an evaluation map 
\begin{equation}
ev_{\ell';l;r}:\operatorname{Map}_{l}(\operatorname{Par}(P)^{\ell'};\mathbb{K})[x_{1},...,x_{r}]\times\operatorname{Nice}(E)^{r}\rightarrow\operatorname{Nice}_{\mathbb{K}}^{\ell'}(E).\label{multivariate_evaluation}
\end{equation}

\begin{itemize}
\item We ask: \emph{can we find subsets $X_{\ell';l}^{r}(P,E)=C_{\ell';l;r}(P)\times\underline{N}^{r}(E)$
of multivariate polynomials and operators such that Lemma} \ref{lemma_4_1_revised}
\emph{holds if we replace the image of $\underline{ev}_{\ell';l}^{C}$
with the image of $ev_{\ell';l;r}$ by $X_{l;r}^{\ell}(P,E)$?}
\end{itemize}
$\quad\;\,$Since we added a new integer index ``$r$'' and since
the desired property holds in the case $r=1$, it is natural to try
to use induction arguments. In an induction argument it is highly
desirable that the induction step can be set in connection with the
base of induction. Recall that for any commutative ring $R$, we have
an isomorphism of graded algebras $R[x,y]\simeq R[x][y]$, so that
$R[x_{1},...,x_{r}]\simeq R[x_{1},...,x_{r-1}][x_{r}]$, which allows
us to work in the nice setting for induction described above. The
obvious idea is to try to take $\underline{N}^{r}(E)=\underline{\operatorname{Nice}}(E)^{r}$,
since for $r=1$ we recover the set of operators used in Lemma \ref{lemma_4_1_revised}.
Furthermore, given $C_{\ell'}(P;\mathbb{K})$ as previously, we can
consider the $C_{\ell';l;r}(P)$ as the subset $C_{\ell';l;r}(P)\subset\operatorname{Map}_{l}(\operatorname{Par}(P)^{\ell'};\mathbb{K})[x_{1},...,x_{r}]$
of multivariate polynomials whose coefficients are in $C_{\ell';l;r}(P)$.
The restriction of $ev_{l;r}$ to those subsets will be denoted by
\[
\underline{ev}_{\ell';l;r}^{C}:C_{\ell';l;r}(P)\times\underline{N}^{r}(E)\rightarrow\operatorname{Nice}_{\mathbb{K}}^{\ell}(E).
\]

We can now restate and prove our main theorem.
\begin{thm}
\label{final_them} Let $S_{1,\varepsilon(\ell)}$ be a homomorphic
or partially homomorphic on parameters NPT, defined in nice background
$\mathcal{NB}_{\ell''}^{\ell}$. Let $C_{\ell''}(P;\mathbb{K})$ be
the domain of a unital functional calculus in $\mathcal{NB}_{\ell''}^{\ell}$.
Then $S_{1,\varepsilon(\ell)}$ emerges from any NPT in the image
of $\underline{ev}_{\ell';l;r}^{C}$, where $\ell''=rl\ell'$.
\end{thm}
\begin{proof}
After all these steps and the discussion above, the proof is almost
straightforward. It is done by induction in $r$. The case $r=1$
is just Lemma \ref{lemma_4_1_revised}. Suppose that the theorem holds
for each $r=1,...,q$ and let us show that it holds for $r=q+1$.
Let $p_{\ell';q+1}^{l}[x_{1},...,x_{r+1}]=\sum_{\vert\alpha\vert\leq l}f_{\alpha}\cdot x^{\alpha}$
be a multivariate polynomial whose coefficients $f_{\alpha}$ belong
to $C_{\ell''}^{\ell'}(P;\mathbb{K})$. From the isomorphism $R[x_{1},...,x_{q+1}]\simeq R[x_{1},...,x_{q}][x_{q+1}]$
one can regard $p_{\ell';q+1}^{l}[x_{1},...,x_{r+1}]$ as a univariate
polynomial $p_{\ell'}^{l_{q+1}}[x_{q+1}]$ for some $1\leq l_{q+1}\leq l$.
Thus, $p_{\ell';q+1}^{l}[x_{1},...,x_{r+1}]=\sum_{i}g_{i;r}[x_{1},...,x_{q}]\cdot x_{q+1}^{i}$,
where $g_{i;r}\in C_{\ell';l_{i};r}(P;\mathbb{K})$ and $1\leq l_{i}\le l$
such that $\sum_{i=1}^{q+1}l_{i}=l$. Let $\Psi_{1},...,\Psi_{q+1}\in\underline{\operatorname{Nice}}(E)$
be right-invertible operators. Notice that 
\begin{equation}
\underline{ev}_{\ell';l;q+1}^{C}(p_{\ell';q+1}^{l}[x_{1},...,x_{q+1}],\Psi_{1},...,\Psi_{q+1})=\sum_{i=1}^{l_{q+1}}\underline{ev}_{\ell';l_{i};q}^{C}(g_{i;q}[x_{1},...,x_{r}],\Psi_{1},...,\Psi_{q})\circ m_{i}[\Psi_{q+1}],\label{theorem_1}
\end{equation}
where $m_{i}[x]$ is the monomial $m_{i}[x]=1\cdot x^{i}$. Let us
define $S_{3_{k},\kappa_{k}}$ as the theory with operators $\Psi_{3_{k},\kappa_{k}}=Q_{k}\circ m_{k}[\Psi_{r+1}]$,
where $Q_{k}=I$ for every $i=1,...,l_{r+1}$. Furthermore, let $S_{2_{j},\delta_{j}}$
be the theory defined by the operators $g_{j;q}[\Psi_{1},...,\Psi_{q}]$.
By the induction hypothesis, $S_{1,\varepsilon(\ell)}$ emerges from
$S_{2_{j},\delta_{j}}$ for every $j$. On the other hand, since $C_{\ell''}(P;\mathbb{K})$
is the domain of a unital functional calculus, from Lemma \ref{lemma_2}
we see that $S_{1,\varepsilon(\ell)}$ and $S_{2_{j},\delta_{j}}$
emerges from $\Psi_{3_{k},\kappa_{k}}$ for every $k$. Thus, the
first two conditions of Lemma \ref{lemma_5} are satisfied. Since
$\Psi_{3_{k},\kappa_{k}}=Q_{k}\circ m_{k}[\Psi_{r+1}]$ it is also
clear that the third hypothesis is also satisfied. The fourth one
follows from Lemma \ref{lemma_2}. Therefore, if the fifth one holds,
then we can apply Lemma \ref{lemma_5} to conclude that $S_{1,\varepsilon(\ell)}$
emerges from (\ref{theorem_1}), concluding the proof. Notice that
the fifth condition, applied to this context, means that 
\[
\Xi_{q+1}=\sum_{i=1}^{l_{q+1}-1}g_{i;q}[\Psi_{1},...,\Psi_{q}]\circ m_{i}[\Psi_{q+1}]
\]
emerges from $\Xi_{q}=g_{l_{q+1};q}[\Psi_{1},...,\Psi_{q}]$. Notice
that $\Xi_{q}=g_{l_{q+1};q}[\Psi_{1},...,\Psi_{q}]=\sum_{j=1}^{l(q)}\Xi_{j;q-1}\circ m_{j}[\Psi_{q}]$,
where $\Xi_{j;q-1}=h_{j;q-1}[\Psi_{1},...,\Psi_{q-1}]$, so that due
to Lemma \ref{lemma_2} $\Xi_{q}$ and $\Xi_{j;q-1}$ emerges from
$m_{j}[\Psi_{q}]$. Thus, from Lemma \ref{lemma_4} if $\Xi_{q+1}$
emerges from $\Xi_{j;q-1}$, then $\Xi_{q+1}$ emerges from $\Xi_{j;q-1}\circ m_{j}[\Psi_{q}]$.
If, in addition $\Xi_{q}$ emerges from $\Xi_{l(q);q-1}$, then we
can use Lemma \ref{lemma_5} to conclude that $\Xi_{q+1}$ emerges
from $\Xi_{q}$, which by the above will finish the proof. Thus, the
real problem is to prove that $\Xi_{q+1}$ emerges from $\Xi_{j;q-1}$
and that $\Xi_{q}$ emerges from $\Xi_{l(q);q-1}$. We can repeat
the argument to show that all we actually need to prove is that $\Xi_{q}$
emerges from $\Xi_{j;q-2}$ and that $\Xi_{q-1}$ emerges from $\Xi_{l(q-1);q-2}$,
where 
\[
\Xi_{q-1}=t_{l(q);q-1}[\Psi_{1},...,\Psi_{q-1}]=\sum_{j=1}^{l(q-1)}\Xi_{j;q-2}\circ m_{j}[\Psi_{q-1}].
\]
Thus, by means of using a reverse induction, now in $q$, we conclude
that it is enough to ensure that $\Xi_{3}$ emerges from $\Xi_{j;1}$
and that $\Xi_{2}$ emerges from $\Xi_{l(1);1}$. But $\Xi_{j;1}$
and $\Xi_{l(1);1}$ are first order univariate polynomials evaluated
in a right-invertible operator. Thus, the result follows from Lemma
\ref{lemma_4_1_revised}.
\end{proof}

\section{\label{sec_examples} Some Examples}

Although this paper is focused on finding general conditions for the
existence of emergence phenomena, in this section we present some
concrete examples aiming to make more clear the real range of our
results. We begin with the basic class of examples, which in the sequence
will be generalized in many directions.
\begin{example}
\label{main_example} Consider the nice background defined by:
\begin{enumerate}
\item \emph{some bounded open set $U\subset\mathbb{R}^{n}$} with the canonical
Riemannian metric (regarded as the spacetime);
\item \emph{the trivial bundle $U\times\mathbb{K}$} (regarded as the field
bundle) with the global sections $C^{\infty}(U;\mathbb{K})$ (regarded
as the fields $\varphi:U\rightarrow\mathbb{K}$);
\item \emph{a number $p\in[1,\infty]$}, corresponding to the integrability
degree;
\item \emph{the graded Hilbert space}\footnote{Recall that the category of Hilbert spaces is closed under direct
sums (coproducts).}\emph{ $W^{p}(U)=\bigoplus_{k\geq0}W^{k,p}(U;\mathbb{K})$}, playing
the role of $\mathcal{H}(E)$;
\item \emph{again the trivial bundle $U\times\mathbb{K}$}, now regarded
as the parameter bundle;
\item \emph{the subspace $cst\subset C^{\infty}(U;\mathbb{K})$ of constant
functions} (viewed as the parameter set), so that $\operatorname{Par}(U\times\mathbb{K})\simeq\mathbb{K}$,
endowed with the canonical algebra structure.
\end{enumerate}
In this nice background, take a NPT defined by:
\begin{enumerate}
\item \emph{degree $\ell=1$}. Thus, $\varepsilon(\ell)$ belongs to $\operatorname{Par}(U\times\mathbb{K})^{\ell}\simeq\mathbb{K}$,
which obviously has degree $\ell=1$ square roots. In the following,
for simplicity we will write just $\varepsilon\in\mathbb{K}$ instead
of $\varepsilon(1)\in\operatorname{Par}(U\times\mathbb{K})^{1}$;
\item \emph{a differential operator $D_{0}:C^{k}(U;\mathbb{K})\rightarrow C^{k-d}(U;\mathbb{K})$
in $U$}, where $0\leq d\leq k$ is its degree. This has its canonical
extension as a bounded operator between Sobolev spaces $\hat{D}_{0}:W^{k,p}(U;\mathbb{K})\rightarrow W^{k-d,p}(U;\mathbb{K})$,
which in turn extends to a bounded operator $\Psi_{0}:W^{p}(U)\rightarrow W^{p}(U)$.
Indeed, let $0_{k;d}:W^{k-d,p}(U;\mathbb{K})\rightarrow W^{k,p}(U;\mathbb{K})$
and $0_{l}:W^{l,p}(U;\mathbb{K})\rightarrow W^{l,p}(U;\mathbb{K})$
be the null operators. Then $\bigoplus_{l\neq k-d}(\hat{D}_{0}\oplus0_{k;d}\oplus0_{l})$
is the desired extension;
\item \emph{the parameterized} \emph{operator $\Psi_{1,\varepsilon}=\varepsilon\Psi_{0}$}.
Notice that the rule $\varepsilon\mapsto\Psi_{1,\varepsilon}$ is
additive and multiplicative, so that the corresponding NPT theory
is homomorphic.
\end{enumerate}
Consider, in addition:
\begin{itemize}
\item \emph{a list $D_{2,1},...,D_{2,r}:C^{\infty}(U;\mathbb{K})\rightarrow C^{\infty}(U;\mathbb{K})$
of smooth linear different operators in $U$ with constant coefficients}.
This means that they have fundamental solutions, i.e, Green functions,
and assume that these are defined in the whole $U$ (this usually
implies constrains on $U$ \citep{right_inverse_const_1,right_inverse_const_2}).
By the above, the operator $D_{2,i}$, with $i=1,...,r$, extend to
bounded operators $\Psi_{2,i}:W^{p}(U)\rightarrow W^{p}(U)$, which
have right-inverses (as pseudo-differential operators) determined
by the Green functions. Let $p^{l}[\Psi_{2,1},...,\Psi_{2,r};\delta]=\sum_{\vert\alpha\vert\leq l}f_{\alpha}(\delta)\Psi_{2}^{\alpha}$,
where $\Psi_{2}^{\alpha}=\Psi_{2,1}^{\alpha_{1}}\circ...\circ\Psi_{2,r}^{\alpha_{r}}$
and $\alpha_{1}+...+\alpha_{r}=\vert\alpha\vert$, be some multivariate
polynomial of degree $l$ whose coefficients are nowhere vanishing
functions $f_{\alpha}:\mathbb{R}\rightarrow\mathbb{R}$ depending
on $\delta$. Denote $D_{2}^{\alpha}=D_{2,1}^{\alpha_{1}}\circ...\circ D_{2,r}^{\alpha_{r}}$.
\end{itemize}
From Theorem \ref{final_them} it then follows that the parameterized
theory with Lagrangian density $\mathcal{L}_{1}(\varphi;\varepsilon)=\varepsilon\varphi^{*}\Psi_{0}\varphi$
emerges from the theory with Lagrangian density 
\[
\mathcal{L}_{2}(\varphi;\delta)=\varphi^{*}p^{l}[\Psi_{2,1},...,\Psi_{2,r};\delta]\varphi=\sum_{\vert\alpha\vert\leq l}f_{\alpha}(\delta)\varphi^{*}\Psi^{\alpha}\varphi.
\]
Furthermore, the theory $\mathcal{L}_{1}(\varphi;\varepsilon)=\varepsilon\varphi^{*}D_{0}\varphi$
also emerges from $\mathcal{L}_{2}(\varphi;\delta)=\sum_{\vert\alpha\vert\leq l}f_{\alpha}(\delta)\varphi^{*}D_{2}^{\alpha}\varphi$.
Here, $\cdot^{*}:\mathbb{K}\rightarrow\mathbb{K}$ is the obvious
involution given by $z^{*}=z$ if $z\in\mathbb{R}$ and $z^{*}=\overline{z}$
if $z\in\mathbb{C}$.
\end{example}
The main conclusion of the above example is the following:$\underset{\underset{\;}{\;}}{\;}$

\noindent \textbf{Conclusion 1.} In a nice open \emph{Euclidean background}
the typical real/complex kinetic scalar theories with scalar parameter
emerge from any multivariate real/complex polynomial scalar theory
with scalar parameter and \emph{constant coefficients}.$\underset{\underset{\;}{\;}}{\;}$

The previous class of examples is the synthesis of how we can use
our abstract and general theorem in order to get concrete information
which is closer to Physics. We note, however, that it can be generalized
in many directions:
\begin{enumerate}
\item \emph{we do not need to assume that the differential operators $D_{2,i}:C^{\infty}(U;\mathbb{K})\rightarrow C^{\infty}(U;\mathbb{K})$
are smooth}. Indeed, we made this assumption only in order to simplify
the notation. In the general situation we could consider functions
$\kappa:[1,r]\rightarrow(0,\infty)$ and $\Delta:[1,r]\rightarrow[0,\infty)$,
with $\Delta(i)\leq\kappa(i)$, and then work on operators $D_{2,i}:C^{\kappa(i)}(U;\mathbb{K})\rightarrow C^{\kappa(i)-\Delta(i)}(U;\mathbb{K})$;
\item \emph{we can work on a more general nice background. }Notice that
the constructions in Example \ref{main_example} are about building
Sobolev spaces and considering extensions of differential operators
to them. For the last one it is implicit the fact that the closure
of smooth functions is equivalent to the Sobolev space. All of this
is true in any compact Riemannian manifold $(M,g)$ and for differential
operators between vector bundles over $M$ \citep{sobolev_manifolds_1,sobolev_manifolds_2,sobolev_manifolds_3}.
This is true even in the case of noncompact manifolds with smooth
boundary, but now under some geometric assumption: if $(M,g)$ has
$k$-bounded geometry, then we have denseness of smooth functions
on $W^{r,p}(M;\mathbb{R})$ for $1\leq r\leq k+2$ \citep{sobolev_manifolds_1,sobolev_manifolds_2,sobolev_manifolds_3}
- see also Footnote \ref{nao_precisa_compact}. Physically, this means
that Conclusion 1 remains true in the case of \emph{Euclidean} backgrounds
with nonzero curvature and of vector or tensor fields instead of scalar
ones;
\item \emph{we can consider other kind of parameters}. In Example \ref{main_example}
we considered $\ell=1$ and $\operatorname{Par}(P)\simeq\mathbb{R}$.
We could considered, more generally, $\operatorname{Par}(P)$ as any
algebra with square roots continuously acting on $\Gamma(E)$, where
$E$ is the field bundle (some vector bundle due the last remark).
Indeed, in this case, due to the denseness of $\Gamma(E)$ on the
Sobolev space $W^{p}(E;\mathbb{K})$ we have an induced action $*\operatorname{Par}(P)\times W^{p}(E;\mathbb{K})\rightarrow W^{p}(E;\mathbb{K})$,
which itself induces an action $*_{B}$ on $B(W^{p}(E;\mathbb{K}))$
given by $(\varepsilon*_{B}\Psi)(\varphi)=\varepsilon*(\Psi\varphi)$.
Furthermore, for every fixed $\Psi_{0}$, the rule $\varepsilon\mapsto\Psi_{1,\varepsilon}=\varepsilon*_{B}\Psi_{0}$
is homomorphic, which is precisely what we need. This includes, for
instance, the case of the nowhere vanishing algebra of positive self-adjoint
real/complex matrices, which can be realized as a nowhere vanishing
subalgebra of $\Gamma(P)$, where $P$ is the algebra bundle $P\times B(\mathbb{K}^{m})$
and $m=\operatorname{rank}(E)$. Physically, this means that in Conclusion
1 we can replace scalar parameter with matrix parameter or more general
and abstract things;
\item \emph{we can consider other kinds of parameterized operators. }In
the above remark, we consider only parameterized operators given by
$\Psi_{1,\varepsilon}=\varepsilon*_{B}\Psi_{0}$, i.e, with a uniform
scaling dependence on $\varepsilon$. More generally, we can take
any representation $\rho_{0}:\operatorname{Par}(P)\rightarrow B(B(W^{p}(E;\mathbb{K})))$,
so that for every bounded operator $\Psi_{0}$ in $W^{p}(E;\mathbb{K})$
we get an induced action $*_{B}$ of $\operatorname{Par}(P)$ on $W^{p}(E;\mathbb{K})$
given by $\varepsilon*_{B}\varphi=[\rho_{0}(\varepsilon)(\Psi_{0})](\varphi)$.
Thus, the rule $\varepsilon\mapsto\Psi_{1,\varepsilon}=\rho_{0}(\varepsilon)(\Psi_{0})$
is homomorphic and defines a nice parameterized operator, as desired;
\item \emph{we can consider NPT of higher degrees}. In all previous points
we worked with $\ell=1$, i.e, with NPT of degree one. But everything
remains true for higher degrees, in virtue of Remark \ref{example_higher_degree_square_roots}.
Physically, we can consider theories which have not a single (scalar,
matrix and so on) fundamental parameter, but many of them, all of
same nature (i.e, all scalar or all matrix, etc.).
\end{enumerate}
$\quad\;\,$If basically everything in Example \ref{main_example}
can be generalized, what are the main difficulties in generalizing
this work? Here are two of them:
\begin{enumerate}
\item \emph{Euclidean background}. Although our results do not require \emph{explicitly}
an Euclidean background, the examples connected with physics (as those
discussed above) depend on such structure. Indeed, this appear in
the construction of the Hilbert space $\mathcal{H}(E)$, which typically
is built from metric structures on the spacetime $M$ and on the field
bundle $E$. Furthermore, we also need to regard differential operators
in $E$ as bounded operators on $\mathcal{H}(E)$, which is done (in
the Euclidean case) by building Sobolev spaces. This is closely related
to the problem of canonical/algebra quantization, where the classical
observables (differential operators) are represented by bounded operators
in a Hilbert space {[}??,??{]}. Looking at this point of view, the
difficulties of avoiding the Euclidean signature are clear. Even so,
if we insist in working on Lorentzian spacetimes $(M,g)$, one can
take induced Riemannian metrics $g_{X}=g+2\omega\otimes\omega$, where
$\omega$ is the 1-form corresponding to a unitary timelike vector
field $X$ in $M$ \citep{lorentz_riemann_1}, and consider our results
on $(M,g_{X})$;
\item \emph{existence of right-inverses}. This is directly related to the
first task. Indeed, recall that in Example \ref{main_example} we
assumed that the differential operators $D_{2,i}$ are of constant
coefficients. Added to a global definition of their Green functions,
this ensured that they are right-invertible. More generally, we could
assume elipticity conditions which are the typical approach to ensure
invertibility \citep{right_inv_1,right_inv_1_2,right_inv_2,right_inv_3,right_inv_4}.
In the \emph{Euclidean }setting, typical operators arising in Physics
are elliptic and additional conditions on their coefficients allows
right-invertibility. On the other hand, in the \emph{Lorentzian }setting,
the same operators become hyperbolic and we cannot use the standard
techniques to ensure right-invertibility.
\end{enumerate}
Thus, the final conclusion is the following:$\underset{\underset{\;}{\;}}{\;}$

\noindent \textbf{Conclusion 2.} In a compact (or noncompact with
bounded geometry) \emph{Euclidean background} the typical real/complex
kinetic field theories with scalar/matrix/etc parameters emerge from
any multivariate real/complex \emph{elliptic} polynomial field theory
with scalar/matrix/etc parameters. Furthermore, our results cannot
be used to directly extend this conclusion to Lorentzian and non-elliptic
settings\footnote{Same kind of difficulties was found in \citep{costello2011renormalization}.}.

\section*{Acknowledgments}

The first author was supported by CAPES (grant number 88887.187703/2018-00).
Both authors would like to thank F\'abio Dadam for reading a preliminary
version of the text. The first author would like to thank Luiz Cleber
Tavares de Brito for introduced him to the subject of the paper many
years ago.

\bibliographystyle{plainnat}
\bibliography{global_EMERGENCE}

\end{document}